\documentclass[runningheads]{llncs}

\usepackage{graphicx}
\usepackage{epsfig,epsf}
\usepackage{epstopdf}
\usepackage{graphics}
\usepackage{array}
\usepackage{amsmath}
\usepackage{amssymb}
\usepackage{comment}
\usepackage{fullpage}
\usepackage{multirow}
\usepackage{enumitem}
\usepackage{bm}

\usepackage{amsmath,amssymb}
\usepackage{mathtools}
\usepackage{setspace}
\usepackage{array}
\usepackage{subcaption}
\usepackage{microtype}
\usepackage{graphics}
\usepackage{cite}

\usepackage[bookmarks=false]{hyperref} 
\usepackage{algorithm}
\usepackage{algpseudocode}
\usepackage{tikz}
\usepackage{adjustbox}

\newtheorem{obs}[theorem]{Observation}

\begin{document}
	\title{Steiner Tree in $k$-star Caterpillar Convex Bipartite Graphs - A Dichotomy\thanks{This work is partially supported by DST-ECRA Project—
			ECR/2017/001442.}}
	\author{Aneesh D H	\inst{1}\and
		A.Mohanapriya \inst{1}\and
		P.Renjith \inst{2}\and
		N.Sadagopan \inst{1}}
	\institute{Indian Institute of Information Technology, Design and Manufacturing,
		Kancheepuram, Chennai. \and
		Indian Institute of Information Technology, Design and Manufacturing, Kurnool.
		\\\email{\{coe16b001,coe19d003\}@iiitdm.ac.in,renjith@iiitk.ac.in,
			sadagopan@iiitdm.ac.in}}		
	\maketitle
	\begin{abstract}
The class of $k$-star caterpillar convex bipartite graphs generalizes the class of convex bipartite graphs.  For a bipartite graph with partitions $X$ and $Y$, we associate a $k$-star caterpillar on $X$ such that for each vertex in $Y$, its neighborhood induces a tree.   The $k$-star caterpillar on $X$ is imaginary and if the imaginary structure is a path ($0$-star caterpillar), then it is the class of convex bipartite graphs.  The minimum Steiner tree problem (STREE) is defined as follows: given a
		connected graph $G=(V,E)$ and a subset of vertices $R \subseteq V(G)$,  the
		objective is to find a minimum cardinality set $S \subseteq V(G)$ such that the
		set $R \cup S$ induces a connected subgraph.  STREE is known to be NP-complete
		on general graphs as well as for special graph classes such as chordal graphs,
		bipartite graphs, and chordal bipartite graphs.   The complexity of STREE in
		convex bipartite graphs, which is a popular subclass of chordal bipartite
		graphs, is open.  In this paper, we introduce $k$-star caterpillar convex bipartite graphs, and show that STREE is NP-complete for $1$-star caterpillar convex bipartite graphs and polynomial-time solvable for $0$-star caterpillar convex bipartite graphs (also known as convex bipartite graphs).   In \cite{muller1987np}, it is shown that STREE in chordal bipartite graphs is NP-complete.   A close look at the reduction instances reveal that the instances are $3$-star caterpillar convex bipartite graphs, and in this paper, we strengthen the result of \cite{muller1987np}.
		\\\textbf{keywords:} $k$-star caterpillar convex bipartite graphs, Steiner tree,  chordal bipartite graphs, convex bipartite graphs.
	\end{abstract}  
	\section{Introduction}
	Many classical subset problems such as vertex cover, independent set and 
	dominating set, have attracted the researchers in the field of theory and
	computing, examining the following aspects: (i) to know whether the problem is
	polynomial-time solvable or NP-complete on general graphs (ii) the status of the
	problem in well-known special graph classes such as chordal graphs, and
	bipartite graphs (iii) if NP-complete on general graphs, then investigate the
	instances generated out of the polynomial-time reduction in an attempt to
	identify easy vs hard instances. (iv) if NP-complete on general graphs,
	investigate the problem from a parameterized complexity perspective with a
	suitable parameter of interest.\\
	The minimum Steiner tree problem (STREE) \cite{garey1979guide} is a classical
	subset problem.  Given an unweighted connected graph $G$ and $R \subseteq V(G)$,
	the problem asks for a minimum cardinality set $S \subset V(G)$ such that the
	set $R \cup S$ induces a connected subgraph.  Subsequently using traversals
	algorithm such as breadth first search or depth first search, one can obtain a
	tree on $R\cup S$, such a tree is known as the Steiner tree for the terminal set
	$R$.   The sets $R$ and $S$ are known as the terminal set and the Steiner set,
	respectively, in the literature.   Interestingly, STREE has applications in road
	construction \cite{grigoreva2015use}, communication networks
	\cite{hwang1995steiner}, computer networks and many more \cite{miehle1958link}. 
	Two of the special cases of STREE are (i) $|R|=2$; in this case, solving STREE
	is equivalent to solving the shortest path problem between the vertices in $R$ 
	(ii) $|R|=|V(G)|$; solving this is equivalent to solving the minimum spanning
	tree problem assuming all edge weights are one. \\
	On the complexity front, STREE is NP-complete on general, and bipartite graphs
	as there is a polynomial-time reduction from the Exact-3-Cover problem
	\cite{hwang1995steiner}.   Further, it is NP-complete on bipartite graphs
	\cite{hakimi1971steiner}, split graphs \cite{white1985steiner}, and chordal
	bipartite graphs \cite{muller1987np}.   For a computational problem known to be
	NP-complete on a graph class, the two possible directions for further research
	are:  (i) study the complexity of the problem in some well-known subclasses of
	the graph class (ii) a closer look at the reduction to understand easy vs hard
	instances.    As part of this paper, we shall take the first direction and
	investigate the complexity of STREE in a subclass of chordal bipartite graphs.  
	The two well-known subclasses of chordal bipartite graphs are convex bipartite
	graphs \cite{damaschke1990domination}, and bipartite distance hereditary graphs
	\cite{d1988distance}.  A bipartite graph $G$ with bipartition $(X,Y)$ is convex, if $X$ can be arranged on a line such that for
	every $y$ in $Y$, its neighborhood consists of consecutive vertices of $X$.
	Interestingly, STREE is polynomial-time solvable in bipartite distance
	hereditary graphs \cite{d1988distance}, however, to the best of our knowledge,
	the complexity of STREE in convex bipartite graphs is open.  In this paper, we
	answer this question and present a polynomial-time algorithm. \\
There is another motivation to this paper. Tree convex bipartite graphs generalize convex bipartite graphs by associating a tree, instead of a path, with one set of the vertices, such that for every vertex in another set, the neighborhood of this vertex induces a subtree. Note that the associated tree or path is imaginary. In this paper, we consider the associated tree to be a special tree, namely {\em $k$-star caterpillar}. We observe that $0$-star caterpillar convex bipartite graphs are the well-known convex bipartite graphs, and hence $k$-star caterpillar convex bipartite graphs generalize the class of convex bipartite graphs. We show that STREE is polynomial-time solvable in $0$-star caterpillar convex bipartite graphs and NP-complete for $1$-star caterpillar convex bipartite graphs. Thus we obtain a P vs NPC dichotomy for STREE in $k$-star caterpillar convex bipartite graphs. \\
There is yet another motivation to this paper.   In \cite{muller1987np}, it is shown that STREE in chordal bipartite graphs is NP-complete.    We observe that the  reduction instances are $3$-star caterpillar convex bipartite graphs, and hence STREE is NP-complete for $3$-star caterpillar convex bipartite graphs.   It is natural to look at a subclass of $3$-star caterpillar convex bipartite graphs where STREE is polynomial-time solvable, and if possible, strengthen the polynomial-time reduction so that we obtain a boundary between P vs NPC instances of STREE in $k$-star caterpillar convex bipartite graphs.   In this paper, we strengthen the result of \cite{muller1987np} and show that STREE is NP-complete for $1$-star caterpillar convex bipartite graphs.
\\Since STREE is a well-studied problem,  we shall highlight some of the important
	results.  On the polynomial time front, STREE is polynomial-time solvable in
	interval graphs \cite{white1985steiner,ramalingam1988unified}, cographs
	\cite{colbourn1990permutation} and $K_{1,4}$-free split graphs
	\cite{renjith2020steiner}.  From the parameterized perspective, STREE is
	fixed-parameter tractable, if the parameter is $|R|$ \cite{dreyfus1971steiner}
	and $W[2]$-hard, if the parameter is $|S|$ \cite{dom2009incompressibility}.  
	The study of STREE is useful in determining the complexity of related problems
	such as connected domination and maximum leaf spanning tree, as observed in
	\cite{white1985steiner,ramalingam1988unified,renjith2020steiner}\\
	{\bf Our Results:} We show that STREE in $1$-star caterpillar convex bipartite graphs is NP-complete by presenting a deterministic polynomial-time reduction from the vertex cover problem.   On $0$-star caterpillar convex bipartite graphs, we show that STREE is in class P.    To present polynomial-time result on $0$-star caterpillar convex bipartite graphs (also known as convex bipartite graphs), the input instances of convex bipartite graphs are
	partitioned into five sets based on the terminal set; $R=X$, $R \subset X$,
	$R=Y$, $R \subset Y$, and $R \cap X \neq \emptyset$ and $R \cap Y \neq
	\emptyset$.   For the first three cases, we present greedy algorithms, and a
	dynamic programming based solution for the other two cases.  
	\section{Graph Preliminaries}
	All graphs considered here are simple, undirected, connected, unweighted graphs.
	We follow the definitions and notation from
	\cite{west2001introduction,golumbic2004algorithmic}.  For a graph $G$, let
	$V(G)$ denote the vertex set and $E(G)$ denote the edge set.  The edge set
	$E(G)=\{\{u,v\} ~|~u$ is adjacent to $v$ in $G$ $\}$.  The open neighborhood of
	a vertex $v$ in $G$ is denoted as $N_G(v) = \{u ~|~ \{u, v\} \in E(G)\}$ and we
	denote the closed neighborhood of a vertex $v$ in $G$ as $N_G[v]=N_G(v)\cup
	\{v\}$.  The degree of a vertex $v$ in $G$ is $d_G(v) = |N_G(v)|$.  We denote by
	$\delta(G) = \min \{d_G(v) ~|~ v \in V(G)\}$.  A vertex $v$ is said to be {\em
		pendant}, if $d_G(v)=1$.  For $V'\subseteq V(G)$, the graph induced on $V'$ is
	represented as $G[V']$.  A bipartite graph is chordal bipartite, if every cycle
	of length strictly greater than four has a chord.  A bipartite graph $G(X,Y)$
	partitioned into $X$ and $Y$ is a convex bipartite graph, if there is an
	ordering of $X=(x_1,\ldots,x_m)$ such that for all $y \in Y$, $N_G(y)$ is
	consecutive with respect to the ordering of $X$, and $G$ is said to have
	convexity with respect to $X$.  For $X=(x_1,\ldots,x_m)$, when we say $x_i \prec
	x_j$, we mean that $x_i$ appears before $x_j$ in the ordering.  Similarly, one
	can define convexity with respect to $Y$.  Convex bipartite graphs can also be
	interpreted as follows:  there exists an imaginary path on $X$ and for each $y
	\in Y$, $N_G(y)$ is an interval (subpath in the imaginary path) in $X$.   For
	every vertex $y \in Y$, $l(y)$ is the least vertex of $X$ adjacent to $y$ and
	$r(y)$ is the greatest vertex of $X$ adjacent to $y$.  We define $T(x_i)$ and a
	vertex $w(x_i) \in N(x_i)$ as follows:  for $i \ge 1$, $T(x_i)=\{y ~|~ y \in
	N(x_i)$, and $r(y)$ is the maximum$\}$, and $w(x_i)$ is an arbitrary vertex of
	$T(x_i)$.   A $k$-star caterpillar, $k \geq 1$, is a tree $T$ where $V(T)=\{x_1,\ldots,x_p \} \cup \{x_{i1},\ldots,x_{ik}\}, 1 \leq i \leq p$ and $E(T)=\{\{x_j,x_{j+1}\} ~|~ 1 \leq j \leq p-1 \} \cup \{
\{x_j,x_{jl}\} ~|~ 1 \leq j \leq p, 1 \leq l \leq k \}$.  A $0$-star caterpillar is a tree $T$ where $V(T)=\{x_1,\ldots,x_p \}$ and $E(T)=\{\{x_j,x_{j+1}\} ~|~ 1 \leq j \leq p-1 \}$.   Equivalently, $0$-star caterpillar is a path on $p$ vertices. \\
	\section{A Polynomial-time Algorithm for STREE in Convex Bipartite graphs}
	We shall present our results by considering possible values for the terminal
	set; towards this, we partition the inputs into five sets.  Throughout this
	paper, we assume convexity on $X$.  We shall next present the solution to STREE
	when $R=X$.
	\subsection{STREE with $R=X$}
	We present a greedy algorithm (Algorithm \ref{algo1}) to solve this case.  Note
	that if $|X|=1$, then the Steiner set is empty.  For $|X|\geq 2$,  using
	convexity on $X$, we identify the vertex $y$ adjacent to $x_1$ such that $r(y)$
	is maximum and we continue from $r(y)$.  Interestingly, this greedy approach is
	indeed optimum, which we establish in this section through a  classical
	cut-and-paste argument \cite{cormen2009introduction}. Let $z_1=x_1$,
	$z_{i+1}=r(w(z_i)), i\geq 1$.
	\begin{algorithm}[H]
		\caption{{\em STREE with $R=X$}}
		\label{algo1}
		\begin{algorithmic}[1]
			\State{{\tt Input:}  A connected convex bipartite graph $G$ with $R=X$.}
			\State{Initialize $i=1$, $z_1=x_1$, $z=r(w(z_1))$}
			\State{Initialize Steiner set $S=\{w(z_1)\}$}
			\While{$z \ne x_m$}
			\State{Update $z_{i+1}=r(w(z_i))$, and $S=S\cup \{w(z_{i+1})\}$}
			\State{$z=r(w(z_{i+1}))$, and $i=i+1$}
			\EndWhile
		\end{algorithmic}
	\end{algorithm}
	\begin{figure}[H]
		\begin{center}
			\begin{tikzpicture}
				\tikzstyle {vertex}=[circle, inner sep=0pt, minimum
				size=0.75ex,draw=black,fill=black]
				\node [vertex](c)[label=left:$x_4$] at (0,1){};
				\node [vertex](d)[label=right:$y_4$] at (2,1){};
				\node [vertex](e)[label=left:$x_3 \text{$=$} r(w(z_1))$] at (0,2){};
				\node [vertex](f)[label=right:$y_3$] at (2,2){};
				\node [vertex](g)[label=left:$x_2$] at (0,3){};
				\node [vertex](h)[label=right:$y_2\text{$=$} w(z_1)$] at (2,3){};
				\node [vertex](i)[label=left:$x_1\text{$=$} z_1$] at (0,4){};
				\node [vertex](j)[label=right:$y_1$] at (2,4){};
				
				\draw[] (i)--(j);
				\draw[] (c)--(d);
				\draw[] (c)--(f);
				\draw[] (e)--(f);
				\draw[] (e)--(d);
				\draw[] (h)--(e);
				\draw[] (g)--(f);
				\draw[] (g)--(h);
				\draw[] (g)--(j);
				\draw[] (i)--(h);
			\end{tikzpicture}
		\end{center}
		\caption{An illustration for the case $R=X$}\label{r=x}
	\end{figure}
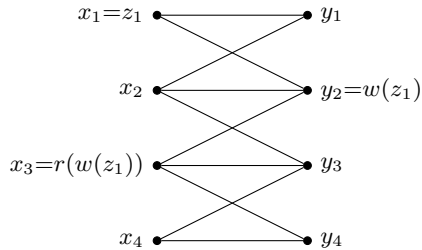
	\noindent An illustration for $R=X$ is given in Figure \ref{r=x} and its trace
	for Algorithm \ref{algo1} is given below.
	\\Note $R=\{x_1,x_2,x_3,x_4\}$.  As part of the initialization, we set
	$i=1,z_1=x_1,z=r(w(z_1))=x_3$, $S=\{y_2\}$.   During the first iteration,  we
	see that $z=x_3 \neq x_4$ is true. 
	Thus, $z_2=x_3 \mbox{ and } S=\{y_2\} \cup \{y_4\}$. Also, $z=x_4,~ i=3$.  Hence
	the solution output by our algorithm is, $S=\{y_2,y_4\}$.  Note that $\{y_1,
	y_3\}$ is also an optimal solution.
	\begin{obs}
		For Algorithm \ref{algo1}, there exists $k$ such that $r(w(z_k))=x_m$, and the
		Steiner set is $S=\{w(z_1),\ldots,w(z_k)\}$.  Thus, Algorithm \ref{algo1}
		terminates.
	\end{obs}
	\begin{theorem}
		Let $G$ be a convex bipartite graph. The set $S$ of Steiner vertices of $G$
		obtained from Algorithm \ref{algo1} is a minimum Steiner
		set.\label{algo1correctness}
	\end{theorem}
	\begin{proof}
		Without loss of generality, we shall arrange the vertices in
		$Y=(y_1,\ldots,y_n)$ such that $S=(y_1,\ldots,y_k)$, are the vertices chosen by
		the Algorithm \ref{algo1} in order.
		We use a binary vector $A=(a_1,\ldots,a_n)$ to represent the output of our
		algorithm such that $a_i=1$, if $y_i\in S$, and $a_i=0$, otherwise.  It follows
		that $a_i=1$, $1\le i\le k$, $a_j=0$, $k+1\le j\le n$.  Let $S'$ denote an
		optimal Steiner set of $G$.  We use a binary vector $B=(b_1,\ldots,b_n)$ to
		represent $S'$ where $b_i=1$, if $y_i\in S'$, and $b_i=0$, otherwise.\\
		Since $S'$ is optimal, $|S'| \leq |S|$.  Further,
		$|B|=\sum\limits_{i=1}^{n}b_i\le \sum\limits_{i=1}^{n}a_i=|A|$.  To show that
		$|S|= |S'|$, we need to prove that $|S|\leq |S'|$, that is $|A| \leq |B|$, we
		need to prove $\sum\limits_{i=1}^{n}a_i\le \sum\limits_{i=1}^{n}b_i$.  We prove
		by mathematical induction on the number of indices $d$ where $A$ and $B$ differ.
		\\\emph{Base case:} when $d=0$, $|A|= |B|$.  Thus, $|A| \leq |B|$.  
		\\\emph{Induction Hypothesis:} Assume that for $d\ge1$, if $A$ and $B$ differ in
		fewer than $d$ positions, then $|A|\le|B|$. 
		\\\emph{Induction Step:} Let the binary vectors $A$, $B$ differ by $d\ge1$
		positions.   Let $j$ be the least index such that $a_j\ne b_j$.  Since $S'$ is
		an optimal solution, it cannot be the case that $a_j=0$ and $b_j=1$.  Therefore,
		$a_j=1$ and $b_j=0$.  Further, $j\le k$.  This implies that $b_i=1$ for $1\le
		i<j$, and $b_j=0$.  Recall that $z_j=r(w(z_{j-1})),~2\leq j \leq m$ and
		$y_j=w(z_{j})$.  Observe that $N(y_{j-1})\cap N(y_j)\neq \emptyset$ and
		$N(y_{j+1})\cap N(y_j)\neq \emptyset$ as $G$ is connected, and by our choice of
		$y_j$, for each $1\leq i \leq j-2,~j+2\leq i \leq k$, $N(y_j)\cap
		N(y_i)=\emptyset$.
		Since $S'$ is an optimal solution, there exists $b_l=1$ with $l>k$ such that
		$\{z_{j},y_l\}\in E(G)$.  If $\{z_{j},y_l\}\notin E(G)$, then feasibility of the
		solution (connectedness) is lost.  That is, the graph induced on $S'\cup X$ has
		vertices $z_j,z_{j+1}$ in different connected components.  This contradicts the
		fact that $S'$ is an optimal Steiner set.  Therefore, $\{z_j,y_l\}\in E(G)$. \\
		Since our algorithm has chosen $y_j$ over $y_l$, it follows that $r(y_l)\le
		r(y_j)$, and $N(\{y_1,\ldots,y_{j-1},y_l\})\subseteq
		N(\{y_1,\ldots,y_{j-1},y_j\})$.  As part of our cut-and-paste argument, we
		modify the vector $B$ to obtain a vector $C=(c_1,\ldots,c_n)$ as follows: 
		$c_i=b_i$, $1\le i\le n$, $i\notin \{j,l\}$, $c_j=1$, $c_l=0$.  It follows that
		the binary vectors $C$ (modified $B$) and $A$ differ in fewer than $d$ positions
		and by the induction hypothesis, $|A|\le|C|$.  Note that $|C|=|B|$.  Thus,
		$|A|\leq |B|$.  We continue this argument, if there is still a mismatch between
		$A$ and $C$, and stop this cut-and-paste argument, when $d=0$.  Thus, $|A|=|B|$,
		and $A$ is also an optimal solution.  This completes the proof of the
		theorem.\qed
	\end{proof}
	\noindent 
	\emph{Remarks:}  The proof is constructive in nature, and given an optimal
	solution, we can obtain another optimal solution by the constructive argument
	mentioned in the proof.
	\subsection{STREE with $R \subset X$}
	We shall now present a greedy algorithm (Algorithm \ref{algo2}) for finding the
	Steiner tree in a convex bipartite graph with $R \subset X$.  When $|X|\leq 2$,
	the Steiner set is empty.  As part of Algorithm \ref{algo2}, we shall consider
	$|X|\geq 3$.  Consider $R=\{z_1,\ldots,z_k\}$, recall that $z_i$ appears before
	$z_{i+1}$ in the ordering of $X$.  We start from $z_1$ and check whether the
	exploration can continue from $r(w(z_1))$ or $z_{j}$, where $z_{j}$ is the
	greatest indexed vertex in $R$ adjacent to $w(z_1)$.   Let $S_1$ be the set of
	vertices chosen by algorithm for obtaining path from $p=r(w(z))$ until
	$z_{j+1}$, and $S_2$ be the set of vertices chosen by algorithm for obtaining
	path from $w(q)$ until $z_{j+1}$, where $q=z_j$.  We choose the minimum out of
	these two subsolutions at each iteration.   This greedy strategy is optimal
	which we establish in this section. 
	\\\begin{minipage}{0.6\textwidth}
		\begin{algorithm}[H]
			\caption{{\em STREE with $R \subset X $}}
			\label{algo2}
			\begin{algorithmic}[1]
				\State{{\tt Input:}  A connected convex bipartite graph $G$ with $R \subset
					X$. }
				\State{Prune the vertices in $X$ less than $z=z_1$}
				\State{Initialize Steiner set $S=\{w(z)\}$, and let $z_j$ be the greatest
					indexed vertex in $R$ adjacent to $w(z)$}
				\While{$j< k$}
				\State{Initialize $p=r(w(z))$, $q=z_j$}
				\If{$p\ne q$}
				\State{Initialize $S_1=\{p\}$, $S_2=\emptyset$ }
				\Else
				\State{Initialize $S_1=\emptyset$, $S_2=\emptyset$ }  
				\EndIf
				\While{$\{z_{j+1},w(p)\}\notin E(G)$}
				\State{$S_1=S_1\cup \{w(p),r(w(p))\}$}
				\State{$p=r(w(p))$}
				\EndWhile
				\While{\{$z_{j+1},w(q)\}\notin E(G)$}
				\State{$S_2=S_2\cup \{w(q),r(w(q))\}$}
				\State{$q=r(w(q))$}
				\EndWhile
				\If{$|S_1|<|S_2|$}
				\State{$S=S\cup S_1\cup \{w(p)\}$; $z=p$}
				\State{Update $z_j$ to be the greatest indexed vertex in $R$ adjacent to
					$w(p)$}
				\Else
				\State{$S=S\cup S_2\cup \{w(q)\}$; $z=q$}
				\State{Update $z_j$ to be the greatest indexed vertex in $R$ adjacent to
					$w(q)$}
				\EndIf
				\EndWhile
			\end{algorithmic}
		\end{algorithm}
	\end{minipage}%
	\begin{minipage}{0.2\textwidth}
		$~~$\\\\\\
		\includegraphics{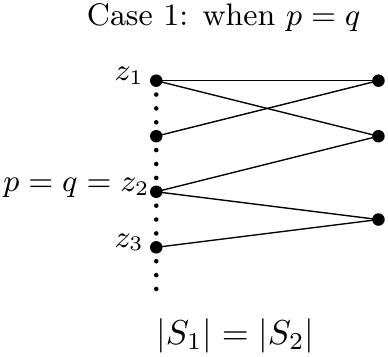}
		\\\\
		\includegraphics{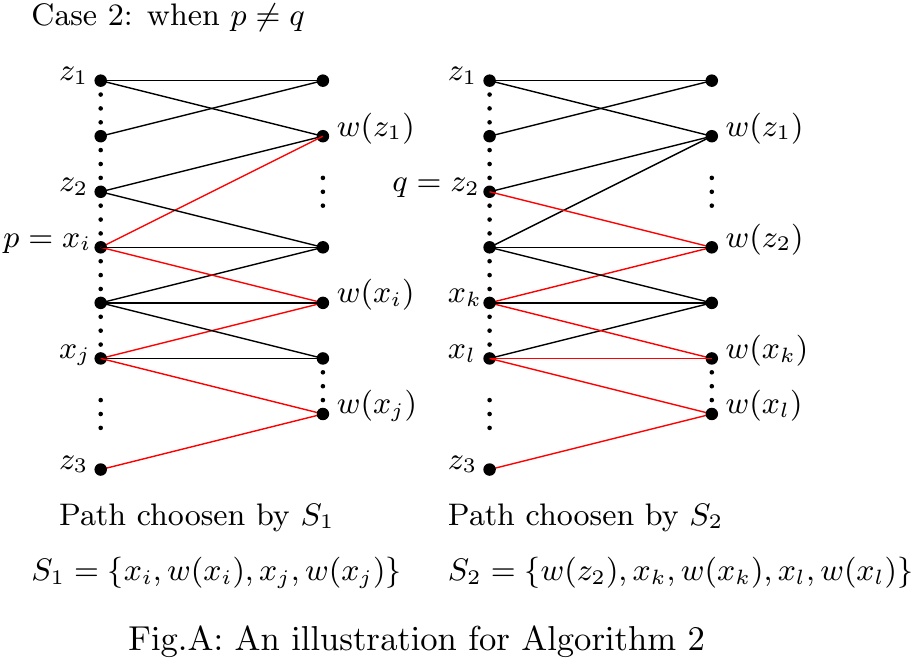}
	\end{minipage}
	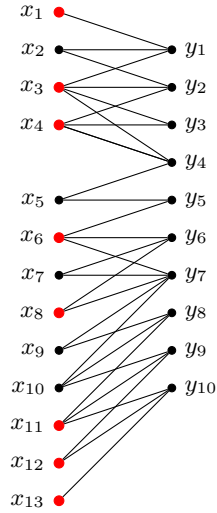
\begin{figure}[H]
		\begin{center}
			\begin{tikzpicture}
				\tikzstyle {vertex}=[circle, inner sep=0pt, minimum
				size=0.75ex,draw=black,fill=black]
				\tikzstyle {vertex1}=[circle, inner sep=0pt, minimum
				size=0.95ex,draw=red,fill=red]
				\node [vertex](15)[label=right:$y_{10}$] at (1.5,-2){};
				\node [vertex](14)[label=right:$y_9$] at (1.5,-1.5){};
				\node [vertex](13)[label=right:$y_8$] at (1.5,-1){};
				\node [vertex](12)[label=right:$y_7$] at (1.5,-0.5){};
				\node [vertex1](7)[label=left:$x_{13}$] at (0,-3.5){};
				\node [vertex1](6)[label=left:$x_{12}$] at (0,-3){};
				\node [vertex1](5)[label=left:$x_{11}$] at (0,-2.5){};
				\node [vertex](4)[label=left:$x_{10}$] at (0,-2){};
				\node [vertex](3)[label=left:$x_9$] at (0,-1.5){};
				\node [vertex1](2)[label=left:$x_8$] at (0,-1){};
				\node [vertex](1)[label=left:$x_7$] at (0,-0.5){};
				\node [vertex](11)[label=right:$y_6$] at (1.5,0){};
				\node [vertex1](a)[label=left:$x_6$] at (0,0){};
				\node [vertex](b)[label=right:$y_5$] at (1.5,0.5){};
				\node [vertex](c)[label=left:$x_5$] at (0,0.5){};
				\node [vertex](d)[label=right:$y_4$] at (1.5,1){};
				\node [vertex1](e)[label=left:$x_4$] at (0,1.5){};
				\node [vertex](f)[label=right:$y_3$] at (1.5,1.5){};
				\node [vertex1](g)[label=left:$x_3$] at (0,2){};
				\node [vertex](h)[label=right:$y_2$] at (1.5,2){};
				\node [vertex](i)[label=left:$x_2$] at (0,2.5){};
				\node [vertex](j)[label=right:$y_1$] at (1.5,2.5){};
				\node [vertex1](k)[label=left:$x_1$] at (0,3){};
				\draw[] (a)--(b);
				\draw[] (c)--(b);
				\draw[] (e)--(d);  
				\draw[] (c)--(d);
				\draw[] (e)--(d);
				\draw[] (e)--(f);
				\draw[] (h)--(e);
				\draw[] (g)--(f);
				\draw[] (g)--(d);
				\draw[] (g)--(h);
				\draw[] (g)--(j);
				\draw[] (i)--(h);
				\draw[] (i)--(j);
				\draw[] (k)--(j);
				\draw (a)--(11);
				\draw (a)--(12);
				\draw (1)--(11);
				\draw (1)--(12);
				\draw (2)--(11);
				\draw (2)--(12);
				\draw (3)--(12);
				\draw (4)--(12);
				\draw (3)--(13);
				\draw (4)--(14);
				\draw (5)--(13);
				\draw (4)--(13);
				\draw (5)--(14);
				\draw (6)--(14);
				\draw (5)--(15);
				\draw (6)--(15);
				\draw (7)--(15);
			\end{tikzpicture}
			\caption{An illustration for $R\subset X$,
				$R=\{x_1,x_3,x_4,x_6,x_8,x_{11},x_{12},x_{13}\}$}
			\label{rsx}
		\end{center}
	\end{figure}
	\noindent An illustration for $R\subset X$ is given in Figure \ref{rsx} and its
	trace for Algorithm \ref{algo2} is given below.
	The terminal vertices are
	$R=\{z_1=x_1,z_2=x_3,z_3=x_4,z_4=x_6,z_5=x_8,z_6=x_{11},z_7=x_{12},z_8=x_{13}\}$.
	Initially, $z=z_1=x_1$, $S=\{y_1\}$, $k=8$.  In Iteration 1; we see that $2<8$,
	$p=r(w(z_1))=x_3$, $q=z_2=x_3$, $p=q$.  Hence $S_1=\emptyset$, $S_2=\emptyset$. 
	At Step $23$, $S$ is updated to $S=\{y_1,y_4\}$, $z=x_3$, $z_j=x_4$. In
	Iteration 2; $3\leq 8$, $p=r(w(z_3))=x_5$, $q=x_4$, $p\neq q$.  By Step $7$, we
	get $S_1=\{x_5\}$, $S_2=\emptyset$.  As per the first while loop; $\{x_6,y_5\}
	\in E(G)$, therefore the condition is false. In the second loop, $\{x_6,y_4\} \notin
	E(G)$, $S_2$ is updated as $S_2=\{y_4,x_5\}$, $q=x_5$. Further in the next
	iteration $\{x_6,y_5\} \in E(G)$, therefore the condition is false and the while
	loop terminates.  We see that Step $19$ is true, $S$ is updated to
	$S=\{y_1,y_4\}\cup \{x_5\}\cup \{y_5\}$.  Further, $z=x_5$, $z_j=z_4=x_6$.  In
	Iteration 3; $4\le 8$, $p=x_6$, $q=x_6$, $p=q$.  Hence by Step $9$,
	$S_1=\emptyset$, $S_2=\emptyset$.  Further in both while loops $\{x_8,y_7\} \in
	E(G)$, therefore conditions are false. At Step 23, $S$ is updated to
	$S=\{y_1,y_4,x_5,y_5\}\cup\{y_7\}$, $z=x_6,z_j=z_5=x_8$.  In Iteration 4; $5 \le
	8$, $p=x_{10}$, $q=x_8$, $p \neq q$ and $S_1=\{x_{10}\}$, $S_2=\emptyset$. 
	Since $\{x_{11},y_9\} \in E(G)$, the while loop condition fails at Step 11 and
	for the other while loop $\{x_{11},y_7\} \notin E(G)$ is true, at Step 15.
	Inside the while loop $S_2$ is updated as $S_2=\{y_7,x_{10}\}$, $q=x_{10}$. At
	Step $23$, $S$ is updated to $S=\{y_1,y_4,x_5,y_5,y_7,x_{10},y_{9}\}$,
	$z_j=z_7=x_{12}$, $z=x_{10}$.  
	In Iteration 5; $7\le 8$, $p=x_{12},q=x_{12}$.  We see that $p=q$, hence
	$S_1=\emptyset$, $S_2=\emptyset$.  In while loops, since $\{x_{13},y_{10}\}\in
	E(G)$, therefore conditions are false. At Step $23$, $S$ is updated to
	$S=\{y_1,y_4,x_5,y_5,y_7,x_{10},y_{9},y_{10}\}$, $z=x_{12}$, $z_j=z_8=x_{12}$.
	In the next iteration, $8<8$ is not true.  Thus, Algorithm \ref{algo2} outputs
	$S=\{y_1,y_4,x_5,y_5,y_7,x_{10},y_{9},y_{10}\}$.
	\begin{obs}
		Let $S$ be any optimal Steiner set.  For each Steiner vertex $y \in Y$, there
		exists at most two Steiner vertices adjacent to $y$ in $S$.
	\end{obs}
	\begin{lemma}
		In Algorithm \ref{algo2}, for each iteration, the difference between $|S_1|$ and
		$|S_2|$ is at most one.
	\end{lemma}
	\begin{proof}
		Let $z$ be the vertex under consideration in $R$ and $z_j$ is the greatest
		indexed vertex in $R$ adjacent to $w(z)$. $p=r(w(z)),~q=z_j$.
		\\\textbf{Case 1: $p = q$.}  Steps 11-18 of Algorithm \ref{algo2} computes $S_1$
		and $S_2$.  Since $p=q$, it implies that $|S_1| = |S_2|$.
		\\\textbf{Case 2:  $p \neq q$.}  Let the path starting from $p$ to $z_{j+1}$ be
		$P_1$ and, the path starting from $q$ to $z_{j+1}$ be $P_2$. Let
		$P_1=(u_1=p=r(w(z)),w(u_1),u_2=r(w(u_1)),w(u_2),\ldots,u_s,w(u_s),z_{j+1})$, and
		
		$P_2=(v_1=q=z_j,w(v_1),v_2=r(w(v_1)),w(v_2),$
		$\ldots,v_t,w(v_t),v_{t+1}=z_{j+1})$. Observe that Steps 11-14 of Algorithm
		\ref{algo2} constructs $P_1$ and updates $S_1$; $S_1=V(P_1)\setminus
		\{w(u_s),z_{j+1}\}$.  Similarly, Steps 15-18 of Algorithm \ref{algo2} constructs
		$P_2$ and updates $S_2$; $S_2=V(P_2)\setminus \{z_j,w(v_t),v_{t+1}\}$.  Since
		$G$ is  convex on $X$, for $1\leq k \leq s,~1\leq l \leq t$,  $u_k\geq v_l$ and
		$\{v_{l+1},w(u_k)\}\in E(G)$.  
		\\\emph{Case:} For some $i\geq 2$, $u_i=v_i$.
		In this case, we observe that $s=t$.  Further, $|S_1|=2s-1, |S_2|=2t-2$.
		Therefore, $|S_1|-|S_2|=1$.
		\\\emph{Case:} For all $i\geq 2$, $u_i>v_i$.
		We observe that $s=t-1$.  Therefore, $|S_2|-|S_1|=1$.  		
		\\By the definition of $p$ and $q$, the case $u_i<v_i$ cannot happen.
		From above two cases we see that $|S_1|$ and $|S_2|$ can differ by at most one.
		\qed
	\end{proof}
	\begin{theorem}\label{algo2correctness}
		Let $G$ be a convex bipartite graph.  The set $S$ of Steiner vertices of $G$
		obtained from Algorithm \ref{algo2} is a minimum Steiner set.
	\end{theorem}
	\begin{proof}
		Without loss of generality, we shall order the vertices in $G$ as
		$\sigma=[v_1,\ldots,v_t]$, $t=|V(G)|$ in a way that $S=(v_1,\ldots,v_l)$ are the
		vertices chosen by Algorithm \ref{algo2} in order.  Note that the ordering
		$\sigma$ is with respect to the ordering of vertices chosen by the algorithm and
		not in accordance with the convex ordering of $X$.  We use a binary vector
		$A=(a_1,\ldots,a_t)$ to represent the output of our algorithm such that $a_i=1$,
		if $v_i\in S$, and $a_i=0$, otherwise.  It follows that $a_i=1$, $1\le i\le l$,
		$a_j=0$, $l+1\le j\le t$.  Let $S'$ denote an optimal Steiner set of $G$. We
		use a binary vector $B=(b_1,\ldots,b_t)$ to represent $S'$ where $b_i=1$, if
		$v_i\in S'$, and $b_i=0$, otherwise.\\
		Since $S'$ is optimal, $|S'|\leq |S|$.  Further,
		$|B|=\sum\limits_{i=1}^{t}b_i\le \sum\limits_{i=1}^{t}a_i=|A|$.  To show that
		$|S'|=|S|$, we need to show that $|S|\leq |S'|$, that is $|A|\leq |B|$.  To show
		that $|A|\leq |B|$, we need to prove $\sum\limits_{i=1}^{t}a_i\le
		\sum\limits_{i=1}^{t}b_i$.  We prove by strong mathematical induction on the
		number of indices $d$ where $A$ and $B$ differ.  
		\\\emph{Base case:} when $d=0$, $|A|=|B|$. Thus, $|A|\le |B|$.  
		\\\emph{Induction Hypothesis:} Assume that for $d\ge1$, if $A$ and $B$ differ
		in less than $d$ positions, then $|A|\le|B|$.  
		\\\emph{Induction Step:} Let the binary vectors $A$, $B$ differ by $d\ge1$
		positions.  Let $j$ be the least index such that $a_j\ne b_j$.  Note that $j\le
		l$, otherwise $S'$ is not optimal.  This implies that $b_i=1$, $1\le i<j$, and
		$b_j=0$.  We consider the following cases to complete our proof.\\
		\emph{Case} $1$: $v_{j-1}\in X$.  Since $S'$ is an optimal solution, there
		exists $b_k=1,k>l$ such that $\{v_{j-1},v_k\}\in E(G)$.  Note that $v_{j}\in Y$
		and $v_j=w(v_{j-1})$.  Similar to the proof of the previous theorem, we modify
		$B$ to obtain a vector $C$ as follows; $C=(c_1,\ldots,c_t)$, $c_i=b_i$, $1\le
		i\le t$, $i\notin \{j,k\}$, $c_j=1$, $c_k=0$.  Note that $|C|=|B|$.  It follows
		that the binary vectors $C$ and $A$ differ in less than $d$ positions and by the
		induction hypothesis, $|A|\le|C|$.
		\\\emph{Case} $2$: $v_{j-1}\in Y$.  We have the following subcases.
		\\\emph{Case} $2.1$: $N(v_{j-1})\cap R=\emptyset$.  Observe that there exists
		$b_k=1,k>l$ such that $\{v_{j-1},x_k\}\in E(G)$.  Note that $v_{j}=r(v_{j-1})$. 
		In this case, an optimal solution with the corresponding vector $C$ is obtained
		from $B$ by changing the values of $b_j,b_k$ as $b_j=1,b_k=0$.  It follows that
		the binary vectors $C,A$ differ in less than $d$ positions and by the induction
		hypothesis, $|A|\le|C|$.  Note that $|C|=|B|$.   Thus, $|A|\le|B|$.
		\\\emph{Case} $2.2$: $N(v_{j-1})\cap R\ne\emptyset$.  Let $z_k$ be the
		greatest indexed vertex in $N(v_{j-1})\cap R$.  If $N(z_{k})\cap
		N(z_{k+1})\ne\emptyset$, then note that $v_j\in Y$.  Observe that there exists
		$b_r=1,r>l$ such that $\{z_{k},v_r\}\in E(G)$.  Note that $v_{j}=w(z_{k})$.  In
		this case, an optimal solution with the corresponding vector $C$ is obtained
		from $B$ by changing the values of $b_j,b_r$ as $b_j=1,~b_r=0$.  Note that
		$|C|=|B|$.  It follows that the binary vectors $C,A$ differ in less than $d$
		positions and by the induction hypothesis, $|A|\le|C|$.  Thus, $|A|\le|B|$.  
		\\If $N(z_{k})\cap N(z_{k+1})=\emptyset$. 
		Note that  $v_j\in X$ or $v_j\in Y$.  Let $v_r$, $r\le l$ be the least vertex
		in $S$ adjacent to $z_{k+1}$.  Note that Steps 11-18 of Algorithm \ref{algo2}
		construct two paths $P_1$ and $P_2$, and choose the minimum out of these two
		paths. Let $P_1=(v_j=r(v_{j-1}),w(v_j),\ldots,v_r,z_{k+1})$, and
		$P_2=(z_k,v_j=w(z_k),\ldots,v_r,z_{k+1})$. 
		\\If $v_j\in Y$, then algorithm chooses $P_2$ as $P_2$ is shortest.  In this
		case, the number of vertices included in $S_2$ by $P_2$ is $r-j+1$ and all these
		vertices appear after $v_{j-1}$ with respect to $\sigma$.  Let $Q$ be the
		vertices chosen by the optimal algorithm to connect $z_k$ and $z_{k+1}$.  Since
		$P_2$ is the shortest path and $Q$ is part of optimal solution, cardinality of
		$Q$ is $r-j+1$.  We now bring our cut-and-paste argument and update $S'$ as
		$S'=(S'\setminus Q) \cup {S_2}$.
		\\Similarly, if $v_j\in X$, then the algorithm chooses $P_1$ as $P_1$ is the
		shortest between $P_1$ and $P_2$.  In this case the number of vertices included
		in $S_1$ by $P_1$ is $r-j+1$.  Let $Q$ be the vertices chosen by the optimal
		algorithm to connect $z_k$ and $z_{k+1}$.  Since $P_1$ is the shortest path and
		$Q$ is part of the optimal solution, the cardinality of $Q$ is $r-j+1$.  We now bring our cut-and-paste argument and update $S'$ as $S'=(S'\setminus Q) \cup {S_1}$.  Let $C$ (modified $B$) be the corresponding binary vector of $S'$.  Note that $|C|=|B|$.  Thus, the binary vectors $C,A$ differ in less than $d$ positions andx	by the induction hypothesis, $|A|\le|C|$.
		This completes the case analysis.  We conclude $|A|=|B|$ and $A$ is also an
		optimal solution.  This completes the proof of Theorem
		\ref{algo2correctness}.\qed
	\end{proof}
	\subsection{STREE when $R=Y$}
	We shall present a greedy algorithm (Algorithm \ref{algo3}) to output a minimum
	Steiner tree when $R=Y$.  Note that if $|Y|=1$, then the Steiner set is empty. 
	Therefore we work with $|Y|\geq 2$.  By definition, for each $y_i \in Y$,
	$N(y_i)=\{x_p,x_{p+1},\ldots,x_q\}$ is an interval.  Further, $l(y_i)=x_p$ and
	$r(y_i)=x_q$.  For all $y_i \in Y$, let $[l_i,r_i]$ represent the interval such
	that $l_i=p$ and $r_i=q$. We arrange the vertices of $Y$ as
	$(y_1,y_2,\ldots,y_n)$ such that for all $i,j$, $1 \leq i < j \leq n$, $r_i\leq
	r_j$.  We use $y_i$ to represent the vertex $y_i\in Y$ as well as the interval
	corresponding to $y_i$.
	\begin{algorithm}[H]
		\caption{{\em STREE when $R=Y$}}
		\label{algo3}
		\begin{algorithmic}[1]
			\State{{\tt Input:}  A connected convex bipartite graph $G$ with $R=Y$.}
			\State{All intervals are unmarked initially, let $|Y|=n$, Steiner set $S$ =
				\{\} }
			\For{$i=1$, $i \leq n$, $i=i+1$}
			\If{$y_i$ is unmarked}
			\State{$S = S \cup \{r_i\}$}     
			\State{Mark all intervals $y_j$ such that $r_i \in N(y_j)$}
			\Else
			\If{$r_i = x_m$} \State{Continue}
			\Else
			\If{there exists a marked $y_j$ such that $r_{i+1} \in N(y_j)$}
			\State{Continue} 
			\Else
			\State{$S = S \cup \{r_i\}$}
			\State{Mark all intervals $y_j$ such that $r_i \in N(y_j)$}
			\EndIf
			\EndIf
			\EndIf
			\EndFor
		\end{algorithmic}
	\end{algorithm}
	\begin{figure}
		\begin{center}
			\includegraphics[scale=0.9]{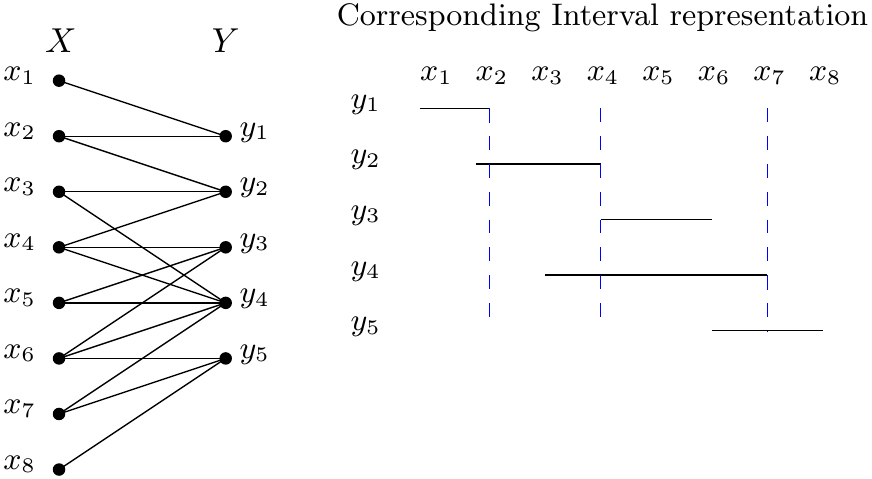}
			\caption{An illustration for $R=Y$}\label{r=y}
		\end{center}
	\end{figure}
	\begin{table}[H]
		\begin{adjustbox}{width=1\textwidth}
			\begin{tabular}{|c|m{15em}|c|c|}
				\hline
				Iteration number & Condition checking and marking status & Update on S &
				Update on marking\\
				\hline 
				&All $y\in Y$ vertices are unmarked initially, $i=1,~S=\{\}$&&\\
				\hline
				1&$1\leq 5$, $y_1$ is unmarked &$S=S\cup \{x_2\}$& Mark $y_1,y_2$\\
				\hline
				2&$2\leq 5$, $y_2$ is marked & $r_i=x_4\neq x_8$, $\nexists y_j$ such that
				$\{x_6,y_j\}\in E(G)$, $S=S\cup \{x_4\}$ & Mark $y_3,y_4$\\
				\hline
				3&$3\leq 5$, $y_3$ is marked& $r_i=x_6\neq x_8$, $\exists y_4$ such that
				$\{y_4,x_7\}\in E(G)$&-\\
				\hline
				4&$4\leq 5$, $y_4$ is marked&$r_i=x_7\neq x_8$,  $\nexists y_j$ such that
				$\{x_8,y_j\}\in E(G)$, $S=S\cup \{x_7\}$&Mark $y_5$\\
				\hline
				5&$5\leq 5$, $y_5$ is marked&$r_i=x_8 = x_8$&-\\
				\hline
				6&$6\leq 5$&-&Thus, $S=\{x_2,x_4,x_7\}$\\
				\hline
			\end{tabular}
		\end{adjustbox}
		\caption{Trace of Algorithm \ref{algo3}}\label{tracer=y}
	\end{table}
	\noindent An illustration and its interval representation is given in Figure \ref{r=y}, and its trace of Algorithm \ref{algo3} is given in Table \ref{tracer=y}.  Let
	$Z=\{z_1,z_2,\ldots,z_p\}\subseteq X$ be the vertices selected by our algorithm
	and $Z'=\{z'_1,z'_2,\ldots,z'_q\}\subseteq X$ be the vertices selected by any
	optimum algorithm.  Note that $z_1\leq z_2 \leq \ldots \leq z_p$.  Further we
	arrange $Z'$ such that $z'_1\leq z'_2 \leq \ldots \leq z'_q$.  For the set
	$\{z_1,z_2,\ldots,z_i\}$,
	$N(\{z_1,z_2,\ldots,z_i\})=\bigcup\limits_{j=1}^iN(z_j)$.
	\begin{theorem}
		\label{trey1}
		For all indices $i\leq q$, the following statements are true:
		\\1. $z'_i\leq z_i$
		\\2. $N(\{z_1,z_2,\ldots,z_i\})\supseteq
		N(\{z^{'}_1,z^{'}_2,\ldots,z^{'}_i\})$ 
	\end{theorem}
	\begin{proof}By mathematical induction on $i$.
		\\\textbf{Base Case:}
		\\Since $z'_1 \leq z'_j$, $j > 1$, we have $\{y_1,z'_1\} \in E(G)$.  Since our
		algorithm has chosen $z_1$, $\{y_1,z_1\} \in E(G)$.
		Therefore, $z'_1 \leq z_1$. The ordering of $Y$ and the convexity
		of $X$ imply that $N(z_1) \supseteq N(z'_1)$.
		\\\textbf{Induction Hypothesis:}
		\\Assume for $ i\geq2$, $z'_{i-1}\leq z_{i-1}$ and
		$N(\{z_1,z_2,\ldots,z_{i-1}\})\supseteq
		N(\{z_1^{'},z_2^{'},\ldots,z_{i-1}^{'}\})$ are true.
		\\\textbf{Induction Step:}
		\\We have to prove that when $ i\geq2$, $z'_i\leq z_i$ and
		$~N(\{z_1,z_2,\ldots,z_{i}\})\supseteq N(\{z_1^{'},z_2^{'},\ldots,z_{i}^{'}\})$.
		\\By the induction hypothesis, we know that up to $i-1$, $z'_{i-1}\leq
		z_{i-1}$ and $N(\{z_1,z_2,\ldots,z_{i-1}\})\supseteq
		N(\{z_1^{'},z_2^{'},\ldots,z_{i-1}^{'}\})$.
		\\By Steps 5 and 14 of Algorithm \ref{algo3}, it is clear that our algorithm
		always includes $z_i=r(y)$ of an interval $y$, hence $z'_i\leq z_i$.
		\\Assume on the contrary, $N(\{z_1,z_2,\ldots,z_{i}\})\nsupseteq
		N(\{z_1^{'},z_2^{'},\ldots,z_{i}^{'}\})$.  Then, there exists an interval $y$
		such that $y\in N(z'_i)$ and $y\notin N(z_i)$.  It is clear that $r(y)\prec z_i$
		(Recall that $r(y)$ appears before $z_i$ in the ordering).  Since for each
		interval $w$ our algorithm includes some $x\in N(w)$ in the solution, it must be
		the case that $y\in N(z_j)$ for some $j$, $1\leq j\leq i-1$ (as illustrated in
		Figure \ref{rey}). This implies that $y\in N(\{z_1,z_2,\ldots,z_{i-1}\})$, which
		is a contradiction.  \\
		\begin{figure}[H]
			\begin{center} 
				\begin{tikzpicture}
					\draw[] (0,0)--(5,0);
					\draw[] (1,0.5)--(4,0.5);
					\draw[] (0,1)--(1.5,1);
					\draw[] (2,1)--(5,1);
					\draw[dashed] (0.5,-0.5)--(0.5,1.5);
					\draw[dashed] (1.5,-0.5)--(1.5,1.5);
					\draw[dashed] (3.5,1.5)--(3.5,-0.5);
					\draw[dashed] (5,1.5)--(5,-0.5);
					\tikzstyle {v1}=[circle, inner sep=0pt, minimum
					size=0.75ex,draw=white,fill=white]
					\node [v1](1)[label=right:$y$] at (4,0.5){};
					\node [v1](2)[label=above:$z^{'}_{i-1}$] at (0.5,1.5){};
					\node [v1](3)[label=above:$z_{i-1}$] at (1.5,1.5){};
					\node [v1](2)[label=above:$z^{'}_{i}$] at (3.5,1.5){};
					\node [v1](2)[label=above:$z_{i}$] at (5,1.5){};
				\end{tikzpicture}
				\caption{Interval representation of G}\label{rey}
			\end{center}
		\end{figure}
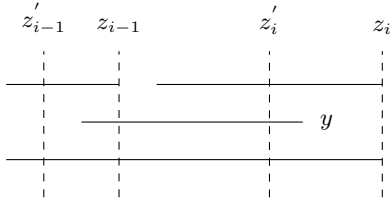
		Therefore,  $N(\{z_1,z_2,\ldots,z_{i}\})\supseteq
		N(\{z_1^{'},z_2^{'},\ldots,z_{i}^{'}\})$. Hence the proof.
		\qed	\end{proof}
	\begin{theorem}\label{trey2}
		For all $k \leq p$, the graph induced on $N[\{z_1,z_2,\ldots,z_{k}\}]$ is
		connected. 
	\end{theorem}
	\begin{proof}
		By mathematical induction on $k$,
		\\\textbf{Base Case:} For $i=1$, by definition $G[N[z_1]]$ is connected.
		\\\textbf{Induction Hypothesis:} Assume that for $i \geq 1$,
		$G[N[\{z_1,z_2,\ldots,z_i\}]]$ is connected.  
		\\\textbf{Induction Step:} We have to prove that when $ i \geq 1,
		~G[N[\{z_1,z_2,\ldots,z_{i+1}\}]]$ is connected.  By our induction hypothesis,
		we know that $G[N[\{z_1,z_2,\ldots,z_i\}]]$ is connected. A vertex $z_{i+1}$ can
		be added to $S$ for two reasons: 
		\\\textbf{Case 1: There exists $y$ such that $r(y)=z_{i+1}$ and $y$ is
			unmarked.} 
		As per Step 6 of our algorithm, $z_{i+1}$ is included in the solution and $y$
		is labelled as marked.  Since $G$ is connected, there exists a marked interval
		$w\in Y$ adjacent to $z_{i+1}$ such that $z_{i+1}=r(w)$ or $z_{i+1}\prec r(w)$. 
		Therefore, graph induced on $N[\{z_1,z_2,\ldots,z_i\}] \cup \{z_{i+1},w,y\}$ is
		connected.
		(Inclusion of $z_{i+1}$ as per the illustration in Figure \ref{rey21})
		\begin{figure}[H]
			\begin{center}
				\begin{tikzpicture}
					\tikzstyle {v1}=[circle, inner sep=0pt, minimum
					size=0.75ex,draw=white,fill=white]
					\draw[] (1,2.5)--(2,2.5);
					\draw[] (2,2)--(3,2);
					\draw[] (2,1.5)--(4,1.5);
					\draw[] (3,1)--(5.5,1);
					\draw[] (5,0.5)--(5.5,0.5);
					\draw[blue,dashed] (2,3)--(2,0);
					\draw[blue,dashed] (4,2)--(4,0);
					\draw[blue,dashed] (5.5,1)--(5.5,0);
					\tikzstyle {v1}=[circle, inner sep=0pt, minimum
					size=0.75ex,draw=white,fill=white]
					\node [v1](1)[label=below:$z_{i-1}$] at (2,0){};
					\node [v1](2)[label=below:$z_{i}$] at (4,0){};
					\node [v1](3)[label=below:$z_{i+1}$] at (5.5,0){};
					\node [v1](3)[label=right:$y$] at (5.5,0.5){};
					\node [v1](3)[label=right:$w$] at (5.5,1){};
				\end{tikzpicture}
				\caption{Interval representation of G} \label{rey21}
			\end{center}
		\end{figure}
		\noindent\textbf{Case 2: There exists $y$ such that $r(y)=z_{i+1}$ and $y$ is
			marked.} 
		It must be the case that there exists an unmarked $w\in Y$ which is adjacent
		to $z_{i+1}$ and not adjacent to $z_i$. To ensure connectedness between
		$G[N[\{z_1,z_2,\ldots,z_i\}]]$ and $w$, our algorithm chooses $z_{i+1}$. Since
		$y$ is ending at $z_{i+1}$, then $y$ is adjacent to one of $z_1,z_2,\ldots,z_i$.
		Therefore, graph induced on $N[\{z_1,z_2,\ldots,z_i\}] \cup \{z_{i+1},w,y\}$
		is connected.(Inclusion of $z_{i+1}$ as per the illustration in Figure
		\ref{rey2})
		\begin{figure}[H]
			\begin{center}
				\begin{tikzpicture}
					\draw[] (1,2.5)--(2,2.5);
					\draw[] (2,2)--(3,2);
					\draw[] (2,1.5)--(4,1.5);
					\draw[] (3,1)--(5.5,1);
					\draw[] (5,0.5)--(5.5,0.5);
					\draw[blue,dashed] (2,3)--(2,0);
					\draw[blue,dashed] (4,2)--(4,0);
					\draw[blue,dashed] (5.5,1)--(5.5,0);
					\tikzstyle {v1}=[circle, inner sep=0pt, minimum
					size=0.75ex,draw=white,fill=white]
					\node [v1](1)[label=below:$z_{i}$] at (2,0){};
					\node [v1](2)[label=below:$z_{i+1}$] at (4,0){};
					\node [v1](3)[label=below:$z_{i+2}$] at (5.5,0){};
					\node [v1](3)[label=right:$y$] at (4,1.5){};
					\node [v1](3)[label=right:$w$] at (5.5,1){};
				\end{tikzpicture}
				\caption{Interval representation of G} \label{rey2}
			\end{center}
		\end{figure}
		Therefore, by both Case 1 and Case 2, $G[N[\{z_1,z_2,\ldots,z_{i+1}\}]]$ is
		connected.\qed
	\end{proof}
	\begin{theorem}
		Algorithm \ref{algo3} outputs a minimum Steiner set, that is $p=q$.
	\end{theorem}
	\begin{proof}
		By Theorem \ref{trey1}, we know that if $i=q$, then
		$N(\{z_1,z_2,\ldots,z_q\})\supseteq N(\{z^{'}_1,z^{'}_2,\ldots,z^{'}_q\})$.  By
		Theorem \ref{trey2}, $N[\{z_1,z_2,\ldots,z_q\}]$ is connected. Hence $p=q$.
		\qed
	\end{proof}
	\noindent\textbf{Time complexity analysis:} For vertices in $Y$, we can maintain an additional data structure so that for each $y\in Y$, $l(y)$ and $r(y)$ can be computed in linear time.  Further, using this data structure and adjacency list of the underlying graph Algorithms 1,2, and 3 can be implemented in $O(m+n)$, linear in the input size.
	\subsection{STREE when $R\subset Y$} \label{rsuby}
	We shall present a dynamic programming based solution for the case $R\subset Y$.
	Let $\sigma=(y_1,y_2,\ldots,y_n)$ be the ordering of vertices in $Y$ satisfying
	the following conditions; for all $i,j$, $1\leq i \le j \leq n$, $y_i$ appears
	before $y_j$ in $\sigma$, if
	\\(i)  $l_i< l_j$, or 
	\\(ii) $l_i= l_j$ and $r_i\geq r_j$.  
	\\We denote by $\sigma(y_i)<\sigma(y_j)$, if $y_i$ appears before $y_j$ in
	$\sigma$.  Similar to Section $3.3$, in this section we work with the underlying
	interval representation of $G$.  Recall that for $y_i\in Y$,
	$N(y_i)=\{x_p,x_{p+1},\ldots,x_q\}$, $l_i=p$ and $r_i=q$.  For $z\in Y$,
	$N(z)=\{x_p,x_{p+1},\ldots,x_q\}$ such that $l(z)=x_p$ and $r(z)=x_q$, we denote
	by $l_z=p$ and $r_z=q$.  Let $R=\{z_1,z_2,\ldots,z_k\}\subset Y$ such that
	$\sigma(z_1)<\sigma (z_2)<\ldots<\sigma(z_k)$. 
	For $z_k\in R$, $x_u=l(z_k)$ and $u=l_k$, and  similarly $x_v=r(z_k)$ and
	$v=r_k$. 
	Let $x_r=l(z_1)$, and $W=\{w_1,w_2,\ldots,w_t\}=\{x_r,\ldots,x_m\},~t=m-r+1$. 
	Note that $x_1,\ldots,x_{r-1}$ is not considered for our discussion, since no
	$z\in R$ is adjacent to $x_1,\ldots,x_{r-1}$. Therefore, we work with $W$ and
	$Y$.  Further, for $y\in Y$, we remove the edges $\{y,x_i\}\in E(G),1\leq i \leq
	r-1$.
	Let $S$ be the set of Steiner vertices required to connect $R$ in $G$.  
	\\We classify $I=[G,R=\{z_1,z_2,\ldots,z_k\}]$ into four equivalence classes
	which are defined as follows;
	\\$E_1=\{I~|~\exists y_c \mbox{ such that }y_c\in N(w_{u-1})\mbox{ and }r_c\geq
	r_k \mbox{, and } \nexists y_d \mbox{ such that } y_c\neq y_d,~ y_d\in
	N(w_{u-1})\mbox{ and }l_k\leq r_d< r_k\}$.
	\\$E_2=\{I~|~\exists y_d \mbox{ such that }y_d\in N(w_{u-1})\mbox{ and }l_k\leq
	r_d< r_k \mbox{, and } \nexists y_c \mbox{ such that }y_c\neq y_d,~y_c\in
	N(w_{u-1}) \mbox{ and } r_c\geq r_k\}$.
	\\$E_3=\{I~|~\exists y_c \mbox{ such that }y_c\in N(w_{u-1})\mbox{ and }r_c\geq
	r_k \mbox{, and }\exists y_d \mbox{ such that }y_d\neq y_c,~y_d\in
	N(w_{u-1})\mbox{ and }l_k\leq r_d< r_k\}$
	\\$E_4=\{I~|~ l_k=1\}$
	\\Informally, $E_1$ considers all those inputs such that in the underlying
	interval representation there exists an interval $y_c$ which overlaps with
	$z_k$, adjacent to $l_{k}-1$ and it ends on or after $r_k$, further, there does
	not exist an interval $y_d$ which overlaps with $z_k$, adjacent to $l_{k}-1$ and
	it ends before $r_k$.
	\\Similarly, $E_2$ considers all those inputs such that in the underlying
	interval representation there exists an interval $y_d$ which overlaps with
	$z_k$, adjacent to $l_{k}-1$ and it ends before $r_k$, further, there does not
	exist an interval $y_c$ which overlaps with $z_k$, adjacent to $l_{k}-1$ and it
	ends on or after $r_k$.
	\\Likewise, $E_3$ considers all those inputs such that in the underlying
	interval representation there exists an interval $y_d$ which overlaps with $z_k$
	and it ends before $r_k$, and there exists an interval $y_c$ which overlaps with
	$z_k$ and it ends on or after $r_k$.
	\\
	In $E_4$, we consider all intervals such that $l_k=1$.  This means each $z_i\in
	R$ is adjacent to $x_1$.
	\\
	Note that, $E_1$, $E_2$, $E_3$, and $E_4$ clearly partitions the set of all
	inputs.
	\\We define an indicator function $b(y)$ for each $y\in Y$ such that:
	\\$$b(y)=
	\begin{cases*}
		1 & \mbox{if $y\in Y\setminus R$}\\
		0 & \mbox{if $y\in R$}
	\end{cases*}$$
	\\Note that $b(z_1)=b(z_2)=\ldots=b(z_k)=0$.
	
	\noindent\textbf{Optimal Substructure Property:}
	We now show that an optimal solution to the Steiner tree problem for the case $R
	\subset Y$ lies within its optimal solutions to subproblems.   Let $T$ be an
	optimal Steiner tree containing $R$.   Clearly, each $z \in R$ appears as a leaf
	in $T$.  Let $w$ be a parent of $z_k$.   If we root the tree at $w$, then both
	left and right subtrees of $w$ must be optimal.  Note that the optimal right
	subtree contains each $z_i,~1 \leq i \leq k-1$ as a leaf.  Further, $w$ is in
	$X$ and $w$ is adjacent to $z_k$ and  $y \in Y$.    Note that $y$ is $z_i,~1
	\leq i \leq k-1$ or $y \in Y \setminus R$.    Moreover, there are many
	candidates for $y$ whose corresponding intervals overlap with $z_k$.   Clearly,
	if all choices of $y$ are considered, then we are sure of obtaining an optimal
	$y$ using which $z_k$ is connected with the rest of vertices in $R$. \\
	Using our optimal substructure, we define a function $F$ which computes the
	minimum number of Steiner vertices required to connect $z_k$ with the rest of
	$R$.   If $z_k$ overlaps with some $z_i,~1 \leq i \leq k-1$, then to obtain an
	optimal solution to the problem we include the appropriate $x \in N_G(z_k) \cap
	N_G(z_i)$ and the optimal solutions obtained from the subproblems.  If $z_k$ has
	no overlap with any $z_i,~1 \leq i \leq k-1$, then $z_k$ overlaps with $y \in Y
	\setminus R$ and there may be many such $y$.  To obtain an optimal solution to
	the problem we include the appropriate $y$ and $x \in N_G(z_k) \cap N_G(y)$ and
	the optimal solutions obtained from the subproblems.    We now present our
	recursive solution to compute $F$.  \\
	\\We define a function $F[u,v]$ which denotes the number of Steiner vertices
	required in $G$ to connect $z_k\in R$ with $z_i\in R,~1\leq i\leq k-1$.
	\\The function $F[u,v]$ for $I$ in $E_1$ or $E_2$ or $E_3$ or $E_4$ is defined
	as follows: $F[u,v]=\min\limits_z f[u,v]$, for each $z\in Y$ such that $u=l_z$
	and $v=r_z$, where $f[u,v]$ is defined as follows:\\
	\\\emph{Case 1:} $I\in E_1$.  Then, $\exists y_c \mbox{ such that }y_c\in
	N(w_{u-1})\mbox{ and }r_c\geq r_k \mbox{, and } \nexists y_d \mbox{ such that }
	y_c\neq y_d,~ y_d\in N(w_{u-1})\mbox{ and }l_k\leq r_d< r_k$.\\
	\\$f[u,v]$= $1+ \min\limits_{y_c} F[p,q],~ 1\leq p \leq u-1,~v\leq q \leq t$\\
	\\\emph{Case 2:} $I\in E_2$.  Then, $\exists y_d \mbox{ such that }y_d\in
	N(w_{u-1})\mbox{ and }l_k\leq r_d< r_k \mbox{, and } \nexists y_c \mbox{ such
		that }y_c\neq y_d,~y_c\in N(w_{u-1}) \mbox{ and } r_c\geq r_k$.\\
	\\$f[u,v]$= $1+b(z_k)+ \min\limits_{y_d} F[p,s],~ 1\leq p \leq u-1,~u\leq s \leq
	v-1$\\
	\\\emph{Case 3:} $I\in E_3$.  Then, $\exists y_c \mbox{ such that }y_c\in
	N(w_{u-1})\mbox{ and }r_c\geq r_k \mbox{, and }\exists y_d \mbox{ such that
	}y_d\neq y_c,~y_d\in N(w_{u-1})\mbox{ and }l_k\leq r_d< r_k$.\\
	\\$f[u,v]$= $\min\{1+ \min\limits_{y_c} F[p,q],1+b(z_k)+ \min\limits_{y_d}
	F[p,s] \},~ 1\leq p \leq u-1,~v\leq q \leq t,~u\leq s \leq v-1$
	\\
	\\\emph{Case 4:} $I\in E_4$\\
	\\$F[u,v]$=1, since $l_k=1$, for each $z_i\in R$, $l_i=1$.\\
	\\Note that, when the input comes from equivalence class $E_i$, there may be
	many identical intervals of type $z$ such that $l_z=u$ and $r_z=v$.  Further, we
	compute $f[u,v]$ for each interval $z$, and $F[u,v]$ is precisely the minimum
	among $f[u,v]$. 
	\\We observe that $F[u,v]$ depends on $F[p,q]$ or $F[p,s]$.  The above
	definition has overlapping subproblems which we shall exploit and present a
	solution using dynamic programming paradigm.  Towards this end we now define a
	recursive solution using which we populate the dynamic programming table in a
	bottom-up.\\
	\\
	\\\textbf{Recursive solution:} \\
	\\\textbf{Base case:}\\
	For $z\in Y$, $l_z=1,~j=r_z$, we define $F[1,j]=\min\limits_z f[1,j],~1\leq
	j\leq t$, the value of $f[1,j]$ is\\
	$f[1,j]=1$, if $z\in Y\setminus R$.\\
	$f[1,j]=0$, if $z\in R$.\\
	$f[1,j]=\infty$, if no such $z$ exists.\\
	\\\\For $2\leq i \leq j \leq t$, for each $z\in Y$, $i=l_z$ and $j=r_z$,
	$F[i,j]=\min\limits_z f[i,j]$, where $f[i,j]$ is
	\\
	\\Case 1: $\exists y_c$ such that $y_c\in N(w_{i-1})$ and $r_c\geq j$, and
	$\nexists y_d$ such that $y_d\neq y_c$, $y_d\in N(w_{i-1})$ and $i\leq r_d< j$\\
	\\$f[i,j]=1+ \min\limits_{y_c} F[p,q],~ 1\leq p \leq i-1,~j\leq q \leq t$
	\\
	\\Case 2: $\exists y_d$ such that $y_d\in N(w_{i-1})$ and $i\leq r_d< j$, and
	$\nexists y_c$ such that $y_c\neq y_d$, $y_c\in N(w_{i-1})$ and $r_c\geq j-1$
	\\\\$f[i,j]= 1+b(z_k)+ \min\limits_{y_d} F[p,s],~ 1\leq p \leq i-1,~i\leq s \leq
	j-1 $
	\\
	\\Case 3: $\exists y_c$ such that $y_c\in N(w_{i-1})$ and $r_c\geq j$, and
	$\exists y_d$ such that $y_d\neq y_c$, $y_d\in N(w_{i-1})$ and $i\leq r_d< j$
	\\\\$f[i,j]= \min\{1+ \min\limits_{y_c} F[p,q], 1+b(z_k)+ \min\limits_{y_d}
	F[p,s]\},~ 1\leq p \leq i-1,~j\leq q \leq t,~i\leq s \leq j-1$\\
	\\The function $F[i,j]=\infty$, if no such $z$ exist.\\
	\\
	\noindent\textbf{Computation of $F[i,j]$:}
	We know that for each interval $z \in Y$, the corresponding function $f[i,j]$ is
	computed.  We compute $f[i,j]$ as per the ordering $\sigma$.   That is, for two
	intervals $y_a$ and $y_b$ such that $\sigma(y_a)<\sigma(y_b)$, then $F[l_a,r_a]$
	is computed first followed by $F[l_b,r_b]$.  We compute $F[i,j]$ for each
	interval $y_a \in Y$ such that $i=l_a$ and $j=r_a$.  The value of $F$ depends on
	the case (the above three cases) in which $y_a$ falls in.  Thus, we consider the
	following three cases and describe how $F$ is computed in each of them. \\
	\\\emph{Case 1:} Note that in this case, we consider all interval $y \in Y$ such
	that $y$ overlaps with $y_a$, and $l_y < l_a$ and $r_y \geq r_a$.  As per
	$\sigma$, for each $y$, we compute $f[l_y,r_y]$.  Since $G$ is connected, $y_a$
	is connected with some $y$.  We examine each $y$ and choose $y$ for which
	$f[l_y,r_y]$ is minimum.  Let $y_{min}=y \in Y$ is such that $f[l_y,r_y]$ is
	minimum.  Clearly, some $x \in N(y_a) \cap N(y_{min})$ is in the solution to
	connect $y_a$ and $y_{min}$.  As part of our approach, we include $x_{l_a} \in
	N(y_a) \cap N(y_{min})$ in our solution.  Thus, we obtain $f[l_a,r_a]=1+
	\min\limits_{y} F[l_y,r_y]$.  The $'1'$ in the expression indicates the
	inclusion of $x_{l_a}$ in the solution, further, it is connected with a $y$
	vertex as indicated by the recursive solution $F[l_y,r_y]$ in the expression. 
	In this case, we do not include $y_a$ in the solution.  Finally, we consider all
	$y_a$ and for each we compute $f[l_a,r_a]$, the minimum over all $f[l_a,r_a]$ is
	precisely $F[l_a,r_a]$.  An illustration is given in Figure \ref{rsy1}.
	\begin{figure}[H]
		\begin{minipage}{0.33\textwidth}
			\includegraphics[scale=0.8]{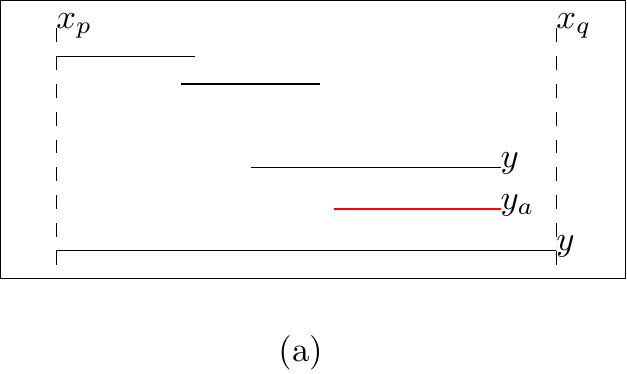}
		\end{minipage}
		\begin{minipage}{0.27\textwidth}
			\includegraphics[scale=0.8]{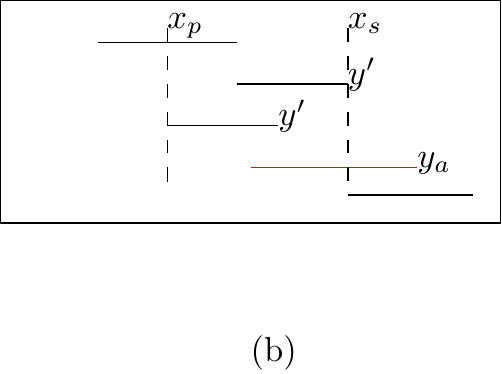}
		\end{minipage}
		\begin{minipage}{0.35\textwidth}
			\includegraphics[scale=0.8]{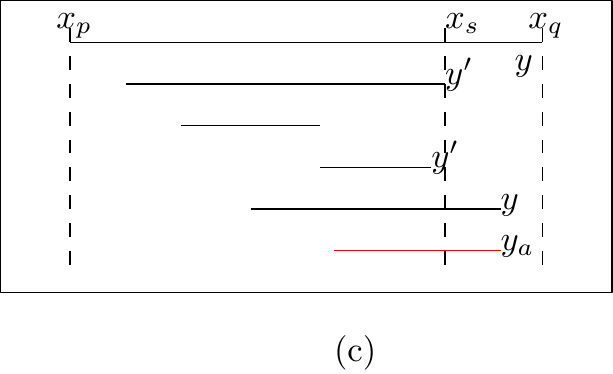}
		\end{minipage}
		\caption{(a) An instance of $G$ for Case 1, (b) An instance of $G$ for Case 2,
			(c) An instance of $G$ for Case 3}\label{rsy1}\label{rsy2}\label{rsy3}
	\end{figure}
	\noindent\emph{Case 2:} 
	Note that in this case, we consider all interval $y' \in Y$ such that $y'$
	overlaps with $y_a$, and $l_y < l_a$ and $r_y < r_a$.   The description for
	computation of $f[l_a,r_a]$ is same as Case 1 and the only change is that
	$f[l_a,r_a]$ includes $x_{l_a}$ and $y_a$ in the solution.  Thus, we obtain
	$f[l_a,r_a]=1+b(y_a) + \min\limits_{y} F[l_y,r_y]$.  The value of $b(y_a)$ is
	'0', if $y_a \in R$, and '1', otherwise.  An illustration is given in Figure
	\ref{rsy2}.
	\\\emph{Case 3:} This case is a blend of Case 1 and 2.  With $y_a$ being the
	reference interval, we find two intervals $y$ and $y'$ in $Y$ such that $y$
	satisfies Case 1 and $y'$ satisfies Case 2.  Accordingly, we compute $F$ for $y$
	and $y'$ and take the minimum of the two.  An illustration is given in Figure
	\ref{rsy3}.
	\\\emph{Case 4:}
	Since $\forall z\in R$, $l(z)=1$, including $w_1$ will connect all the vertices
	in $R$.  Hence $F[i,j]=1$.
	\\
	\\\textbf{Overlapping subproblems in $F$:}
	\\Consider the subproblems $F[p,q],~F[a,b],~F[c,d]$ such that $p<a$, $p<c$ and
	$q\geq a,~q\geq c$.  Since we compute $F[i,j]$ in bottom-up and $p<a$, $p<c$,
	$F[p,q]$ is computed before $F[a,b]$ and $F[c,d]$.  We observe that $F[p,q]$ is
	a  subproblem in $F[a,b]$ and $F[c,d]$.  We compute $F[p,q]$ once and reuse the
	solution when it is referred again.  Therefore each subproblem is computed
	exactly once.\\
	\textbf{Computing the optimal Steiner set using $F$:}\\
	Using $F[u,v]$, we construct the solution set $S$ in a bottom up starting from
	minimum $f[u,v]$.  If suppose, $f[u,v]$ is updated due to $f[p,q]$ of $F[p,q]$,
	then include the vertex $w_u$, $w_p$ in $S$, and also include the corresponding
	$y$ vertex in $S$.  We continue this process until we reach either some
	$f[1,j],~1\leq j \leq n$ or $f[p,q]$ such that $p=l_{{z_1}}$.  If there exist
	$z_i$ such that $N(z_i)\cap S=\emptyset$ and there does not exist $z_j$ such
	that $\sigma(z_i)<\sigma(z_j)$, $N(z_i)\cap N(z_j)\neq \emptyset$, then include
	$l(z_i)\in R$ (An illustration for inclusion of $z_i$ in $S$ is in Figure
	\ref{$z_i$}).
	The vertices $S \setminus R$ are the desired Steiner vertices of $G$ for the
	terminal set $R$.
	\\An illustration for inclusion of $z_i$ in $S$, for $z_3,z_4,z_5$ we include
	$l(z_4),l(z_5)$ in $S$.
	\begin{figure}[H]
		\begin{center}
			\includegraphics[scale=0.8]{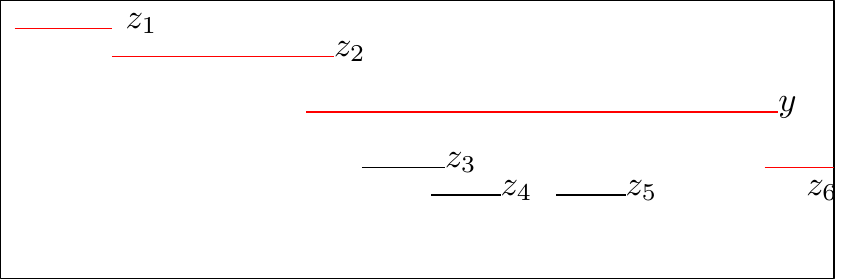}
			\caption{An illustration for inclusion of $z_i$ in $S$.}\label{$z_i$}
		\end{center}
	\end{figure}
	\noindent\textbf{Pseudo code to compute $F$:}
	\begin{algorithm}[H]
		\caption{{\em Computing $F$}}
		\label{algo4}
		\begin{algorithmic}[1]
			\State{{\tt Input:}  A connected convex bipartite graph $G$ with $R\subset
				Y$.}
			\State{$oldFvalue=\infty$ , $newFvalue=\infty$}
			\For{each $z\in Y$ such that $l_z=1$ and $r_z=j$}
			\If{$z\in Y\setminus R$}
			\EndIf
			\State{$F[1,j]=1$}
			\If{$z\in R$}
			\State{$F[1,j]=0$}
			\EndIf
			\EndFor
			\For{$j=1,j\leq m,j=j+1$}
			\If{there does not exists a vertex such that $l_z=1$ and $r_z=j$}
			\State{$F[1,j]=\infty$}
			\EndIf
			\EndFor
			\For{$i=2,i\leq t,i=i+1$}
			\For{each $z\in Y$ such that $l(z)=x_i$}
			\State{let $j=r(y)$}
			\If{$\exists y_c$ such that $y_c\in N(w_{i-1})$ and $r_c\geq j$, and $\exists
				y_d$ such that $y_d\neq y_c$, $y_d\in N(w_{i-1})$ and $i\leq r_d< j$}
			\State{$f[i,j]= \min\{1+ \min\limits_{y_c} F[p,q], 1+b(z_k)+ \min\limits_{y_d}
				F[p,s]\},~ 1\leq p \leq i-1,~j\leq q \leq t,~i\leq s \leq j-1$}
			\ElsIf{$\exists y_d$ such that $y_d\in N(w_{i-1})$ and $i\leq r_d< j$}
			\State{$f[i,j]= 1+b(z_k)+ \min\limits_{y_d} F[p,s],~ 1\leq p \leq i-1,~i\leq s
				\leq j-1 $}
			\Else
			\State{$f[i,j]=1+ \min\limits_{y_c} F[p,q],~ 1\leq p \leq i-1,~j\leq q \leq
				t$}
			\EndIf
			\State{$newFvalue=\min\{f[i,j],oldFvalue\}$}
			\State{$oldFvalue=newFvalue$}
			\EndFor
			\State{$F[i,j]=newFvalue$}
			\EndFor
		\end{algorithmic}
	\end{algorithm}
	\noindent\textbf{Time complexity of the function $F$:}
	\\As the range of $i$ ($j$) is $1$ to $m$ and for each $z \in Y$, we compute the
	function $F$,  the number of subproblems $F$ created by our dynamic programming
	is at most $O(m^2)$.  Further, the number of updates on $F$ is $O(m^2n)$.  
	Thus, Steiner tree when $R \subset Y$ runs in $O(m^2n)$, polynomial in the input
	size.
	\begin{theorem}
		For a convex bipartite graph $G$ and a terminal set $R \subset V(G)$, the
		Steiner set output by our algorithm is an optimal Steiner set.
	\end{theorem}
	\begin{proof}
		Let $R=\{z_1,\ldots,z_k\}$.  With $z_k$ as the reference interval, we first
		identify the equivalent class in which $G$ falls into.  Further, we compute
		$F[i,j]$ in a specific order so that solutions to subsubproblems are made
		available to the subproblems and further to the actual problem.   Thus, the optimal
		solution to $F[u,v]$ is obtained by considering all optimal subproblems. 
		Therefore, the set output by our algorithm is an optimal Steiner set.
		\qed
	\end{proof}
	\begin{figure}[H]
		\begin{center}
			\includegraphics[scale=0.9]{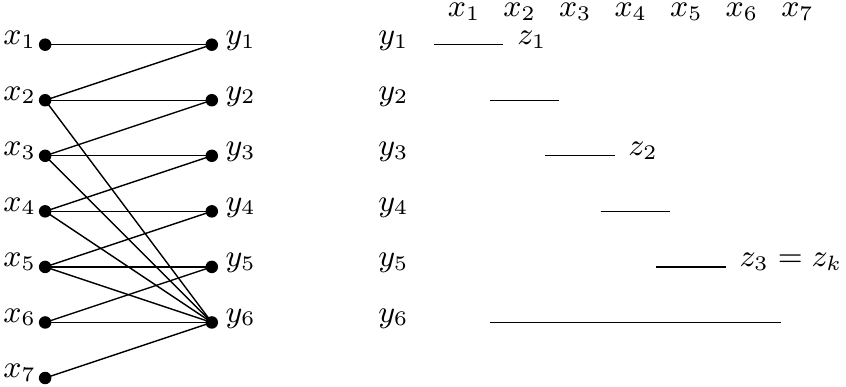}
			\caption{An illustration for $R\subset Y$}\label{rsy}
		\end{center}
	\end{figure}
	We now trace our algorithm for the example given in Figure \ref{rsy}. \\
	Base case: $F[1,2]=0$.
	\\At $i=2$, there exist two intervals $y_2,y_6$ which starts at $x_2$.  The $F$
	values computed are:
	\\for $y_2$, $F[2,3]=1+b(z)+\min{(F[1,2],F[1,3])}=1+1+0=2$
	\\for $y_6$,
	$F[2,7]=1+b(z)+\min{(F[1,2],F[1,3],F[1,4],F[1,5],F[1,6],F[1,7])}=1+1+0=2$.
	\\At $i=3$, there exists an interval $y_3$ which starts at $x_3$. The function
	$F$ is computed for
	\\$y_3$ is
	$F[3,4]=1+\min{(F[1,3],F[1,4],F[1,5],F[1,6],F[1,7],F[2,3],F[2,4],F[2,5],F[2,6],F[2,7])}=1+2=3$.
	\\At $i=4$, there exists an interval $y_4$ which starts at $x_4$. The function
	$F$ computed for
	\\$y_4$ is
	$F[4,5]=1+\min{(F[1,4],F[1,5],F[1,6],F[1,7],F[2,4],F[2,5],F[2,6],F[2,7],F[3,4],F[3,5],F[3,6],F[3,7])}=1+2=3$.
	\\At $i=5$, there exists an interval $y_5$ which starts at $x_5$. The function
	$F$ computed for
	\\$y_5$ is
	$F[5,6]=1+\min{(F[2,5],F[2,6],F[2,7],F[3,5],F[3,6],F[3,7],F[4,5],F[4,6],F[4,7])}=1+2=3$.
	\\\textbf{Constructing an optimal solution:} For this input instance $F[u,v]$ is
	$F[5,6]$.  The $'1'$ in the expression $F[5,6]$ indicates the inclusion of $x_5$
	in $S$, $S=x_5$.  Since the value of $F[5,6]$ is updated due to $F[2,7]$, we
	next consider $F[2,7]$.  Now in the expression $F[2,7]$, $'1'$ indicates the
	inclusion of $x_2$ in $S$, and $b(z)=1$, which refers to the inclusion of $y_6$
	in $S$, $S=\{x_5,x_2,y_6\}$.  On the similar line $F[2,7]$ is updated due to
	$F[1,2]$.  Thus include $x_2$ in $S$.  Finally, there exists an interval $z_2$
	such that $N(z_2)\cap S=\emptyset$, hence include $x_3$ in $S$,
	$S=\{x_5,x_2,y_6,x_3\}$.  Therefore, the Steiner vertices of $G$ is
	$\{x_2,x_3,y_6,x_5\}$.
	\subsection{STREE when $R \cap X\neq \emptyset$ and $R \cap Y\neq \emptyset $} \label{allinone}
	Let $R=\{z_1,\ldots,z_l\}$ such that $R \cap X=\{z_1,\ldots,z_k\},~1\leq k <l$, and $R \cap Y=\{z_{k+1},\ldots,z_l\}$.  To describe the
	solution for this case,
	we transform the graph $G(X,Y)$ to $G^*(X^*,Y^*)$ such that
	$X^*=X,Y^*=Y\cup W$, $ W=\{w_i~|~z_i\in R\cap X,~1\leq i \leq k$\} and
	$E(G^*)=E(G)\cup \{\{w_i,z_i\}|~w_i\in Y^*,~z_i\in R\cap X,~1\leq i \leq k\}$. 
	Note that each $w_i$ is a pendant vertex in $G^*$.  Observe that the convex
	ordering of $X^*$ is same as $X$, and for each $y\in Y^*$, $N_{G^*}(y)$ is consecutive
	with respect to the ordering of $X^*$.  Therefore, $G^*$ is a convex bipartite
	graph.  Moreover, this transformation is a solution preserving transformation.  That is, using the transformed graph $G^*$, we obtain a solution to STREE in $G$.  In particular
	$(G,R)$ is mapped to $(G^*,R^*)$ such that $R^*=(R\cap Y)\cup W$.
	\\Clearly, $R^*\subset Y^*$.  Using the dynamic programming presented in Section \ref{rsuby}, we solve $(G^*,R^*)$, and let $S^*$ be the solution to $G^*$.  Note that since each $w_i$ is pendant and $w_i\in R^*$.  Hence no $w_i$ is in $S^*$.   Therefore $S^*$ is also a solution in $G$.
	\noindent\emph{Remarks:} To solve STREE in convex bipartite graphs, it is enough to consider the case STREE when $R \subset Y$, and all other cases can be transformed to an instance of $R \subset Y$ using the construction given in Section \ref{allinone}.   It is important to highlight that $R \subset Y$ algorithm runs in time $O(m^2n)$, whereas all other greedy algorithms run in linear time.     
\section{Hardness Result: STREE in 1-star caterpillar convex bipartite graphs}
In this section, we show that the STREE is  NP-complete for 1-star caterpillar convex bipartite graphs by giving a polynomial-time reduction from the vertex cover problem on general graphs.  A 1-star caterpillar is a tree $T=(V,E)$ where $V(T)=\{v_1,\ldots,v_n,z_1,z_2,\ldots,z_n\}$ and $E(T)=\{\{v_i,z_i\}~|~1\leq i \leq n\}\cup \{\{v_i,v_{i+1}\}~|~1\leq i \leq n-1\}$.  Note that $(v_1v_2\ldots v_n)$ is the path of the caterpillar (also known as the backbone of the caterpillar), and $\{z_1,z_2,\ldots,z_n\}$ are {\em pendants} of the caterpillar.  The decision versions of STREE and the vertex
	cover problem are defined below:
	\\\textbf{The Steiner tree problem (STREE)}
	\\Instance: A graph $G$, a terminal set $R\subseteq V(G)$,
	a non-negative integer $k$.
	\\Question: Does there exist a Steiner set $S\subseteq V(G)$ such that $G[R\cup
	S]$ is connected and $|S| \leq k$?
	\\\textbf{The Vertex Cover problem (VC)}
	\\Instance: A graph $G$, a non-negative integer $k$.
	\\Question: Does there exist a vertex cover $S\subseteq V(G)$ such that for each
	edge $e=\{u,v\}\in E(G)$, $u\in S$ or $v\in S$ and $|S| \leq k$?
	\begin{theorem} \label{redn}
		STREE is NP-complete on 1-star caterpillar convex bipartite graphs.
	\end{theorem}
	\begin{proof}
		\textbf{STREE is in NP:} Given an input instance $(G,R,k)$ of STREE and a certificate set $S\subseteq V(G)$, whether $S$ is a Steiner set of cardinality at most $k$ can be verified in polynomial time as the connectedness of $G[R\cup S]$ can be verified in polynomial time by using standard graph traversal algorithms \cite{cormen2009introduction}.
		\\\textbf{STREE is NP-hard:} It is known from \cite{cormen2009introduction} that VC on general graphs is NP-complete and this can be reduced in polynomial time to STREE in 1-star caterpillar convex bipartite graphs using the following reduction algorithm.   We map an instance $(G,k)$ of VC on general graphs to the corresponding instance $(G^*,R,k'=k)$ of STREE as follows: $V(G^*)=V_1\cup V_2\cup V_3$, \\
		$V_1=\{x_i~|~v_i\in V(G)\}$, \\ $V_2=\{y_{i1},y_{i2}~|~e_i\in E(G)\}$, \\
		$V_3=\{z_{i1},z_{i2}~|~e_i\in E(G)\}$. \\  We shall now describe the edges of
		$G^*$, \\
		$E(G^*)=E_1 \cup E_2, \\
E_1= \{\{y_{i1},x_k\},\{y_{i1},x_l\},\{y_{i2},x_k\},\{y_{i2},x_l\},~|~e_i=\{v_k,v_l\}\in
		E(G),~x_k,x_l\in V_1,~y_{i1},y_{i2}\in V_2,~1\leq i \leq m,~1\leq k \leq
		n,~1\leq l \leq n\} \\
E_2=\{\{x,z_{i1}\}, \{x,z_{i2}\}~|~x\in
		V_1,~z_{i1},z_{i2}\in V_2 ,~1\leq i \leq m\}$.
		\\We define $X^*=V_2\cup V_3$, $Y^*=V_1$, and imaginary 1-star caterpillar $T$ on $X^*$ is defined with $V_3$ as the backbone and $V_2$ as the pendant vertex set.  That is, $V(T)=X^*$ and $E(T)=\{\{y_{11},y_{12}\}, \{y_{12},y_{21}\}, \{y_{21},y_{22}\}, \ldots, \{y_{m1},y_{m2}\}\} \cup \{\{y_{i1},z_{i1}\}, \{y_{i2},z_{i2}\}  ~|~ 1 \leq i \leq m \}$. 
 \\ An example is illustrated in Figure \ref{comb},  the vertex cover instance $G(V,E)$ with $k=2$ is mapped to STREE instance of 1-star caterpillar convex bipartite graph $G^*(V^*,E^*)$ with $R=\{y_{11},y_{12},y_{21},y_{22},y_{31},y_{32},z_{11}\}$, $k'=2$. 
		\begin{figure}[H]
			\begin{center}
				\includegraphics[scale=0.95]{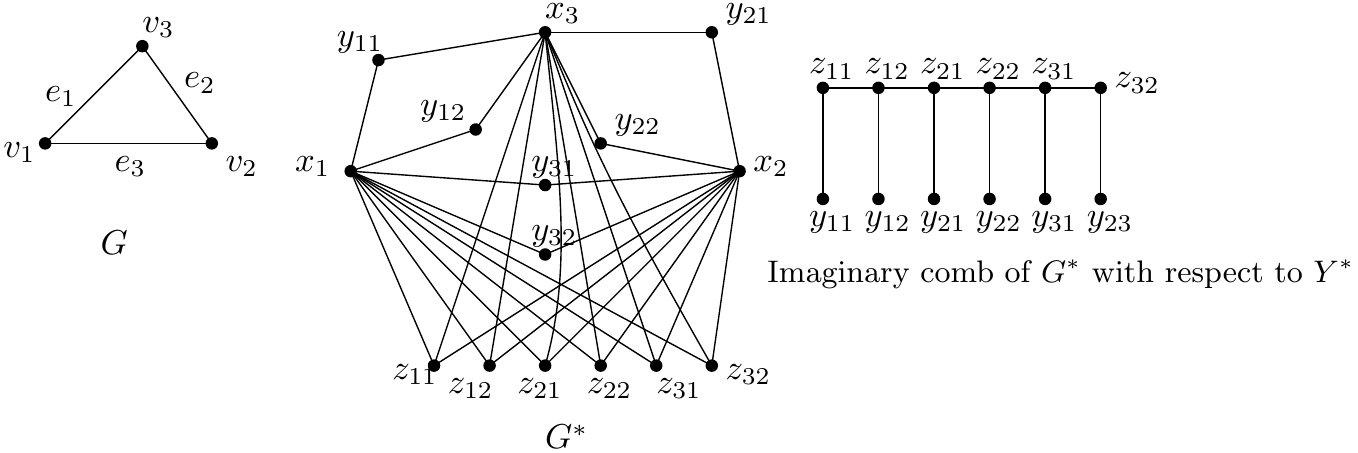}
				\caption{An example: VC reduces to STREE.}\label{comb}
			\end{center}
		\end{figure}
		\begin{claim}
			$G^*$ is a 1-star caterpillar convex bipartite graph.
		\end{claim}
		\begin{proof}
			By construction, $T$ is a 1-star caterpillar on $X^*$.   Each $x_i \in V_1$ (or $x_i \in Y^*$ ) is adjacent to all vertices in $V_3$ and also to each element in some subset $Y' \subseteq V_2$.    Therefore, for each $x_i$,  $N(x_i)$ is a subtree in $T$.  Hence $G^*$ is a 1-star caterpillar convex bipartite graph.
		\end{proof}
		\begin{claim}
			$(G,k)$ has a vertex cover with at most $k$ vertices if and only if
			$(G^*,R=\{y_{i1},y_{i2}~|~1\leq i\leq m\}\cup \{z_{11}\},k'=k)$ has a Steiner tree of size at most $k'=k$ Steiner vertices.
		\end{claim}
		\begin{proof}
			\emph{(Only if)} Let $V'=\{v_i~|~1\leq i \leq k\}$ is a vertex cover of size $k$ in $G$.   Then we construct the Steiner set $S$ of $G^*$ for $R=\{y_{i1},y_{i2}~|~1\leq i\leq m\}\cup \{z_{11}\}$ as follows $S=\{x_i ~|~ 1 \leq i \leq k, ~v_i \in V', ~x_i \in V(G^*)\}$.
			Indeed, for any edge $e_i=\{v_k,v_l\}\in E(G)$, $v_k$ or $v_l$ in $V'$.   Then by our construction, we know that $y_{i1}$ and $y_{i2}$ are adjacent to $x_k$ and $x_l$, and  $x_k$ or $x_l$ is in $S$.  So each vertex in $\{y_{i1},y_{i2}~|~1\leq i \leq m\}$ is adjacent to at least one vertex in $S$.   Further,  by our construction, each vertex in $V_1$ is adjacent to each vertex in $V_3$.  Hence $S \cup R$ induces a connected subgraph in $G^*$.
			\\\emph{(If)} For $R$ in $G^*$, let $S=\{x_i~|~1\leq i \leq k'\}$ is a Steiner set of $G^*$ of size $k'$.  Then, we construct the vertex cover $V'$ of size $k$ in $G$ as follows; $V'=\{v_i~|~x_i \in S,~v_i\in V(G),~1\leq i\leq k'\}$.   We now claim that $V'$ is a vertex cover in $G$.   Suppose that there is an edge $e_i=\{v_k,v_l\}\in E(G)$ for which neither $v_k$ nor $v_l$ is in $V'$.   This implies that neither $x_k$ nor $x_l$ is in $S$.   Since $R$ contains $y_{i1},y_{i2}$, it follows that $N(y_{i1}) \cap S=\emptyset$ and $N(y_{i2})\cap S=\emptyset$.   Further, $S$ is not a Steiner set.  A contradiction.   Thus $V'$ is a vertex cover of size $k$ in $G$.\qed
		\end{proof}
	Therfore, we conclude STREE on $1$-star caterpillar convex bipartite graphs is NP-complete.	\qed
	\end{proof}
	\begin{corollary}
		STREE is NP-complete on $k$-star caterpillar convex bipartite graphs, $k \geq 1$.    Further, STREE is NP-complete on tree convex bipartite graphs.
	\end{corollary}
	\begin{proof}
	Since the class of $1$-star caterpillar convex bipartite graphs is a special case of $k$-star caterpillar convex bipartite graphs, $k \geq 1$, and the fact $k$-star caterpillar convex bipartite graphs are a subclass of tree convex bipartite graphs, this result follows from Theorem \ref{redn}.
	\end{proof}
\noindent Remark: In \cite{muller1987np}, it is shown that STREE in chordal bipartite graphs is NP-complete, and it is important to highlight that it is a 3-star caterpillar convex bipartite graph.  Hence STREE is NP-complete for 3-star caterpillar convex bipartite graph.
\\In Theorem \ref{redn}, we strengthen the result of \cite{muller1987np} by establishing the NP-complete result for 1-star caterpillar convex bipartite graphs.\\
We shall next present two applications of our result. We use STREE in convex bipartite graphs as a framework and solve (a) STREE in intervals graphs, and (b) Domination in convex bipartite graphs. To the best of our knowledge STREE in interval graphs is open, and the study of domination in convex bipartite graphs is already reported in \cite{damaschke1990domination}. 	
\section{An Application: STREE in Interval graphs}
It is known from \cite{white1985steiner} that STREE on chordal graphs is NP-complete.   The class of interval graphs is a popular subclass of chordal graphs on which STREE is open.  In this paper, we present a polynomial-time algorithm for STREE in interval graphs using STREE in convex bipartite graphs as a black box.   In particular, we invoke STREE in convex bipartite graphs with $R=Y$ algorithm to solve STREE in interval graphs.   It is important to highlight that the interval representation used in Section \ref{} for $R=Y$ is on the integer line. \\
A graph $G$ is an interval graph if there exists a $1$-$1$ correspondence between its vertices and a set of intervals on the real line such that two vertices are adjacent if and only if the corresponding intervals have a nonempty intersection \cite{booth1975linear}.     For an interval graph $G$ with $V(G)=\{v_1,v_2,\ldots,v_n\}$ and $E(G)=\{e_1,\ldots,e_m\}$, we associate an interval $I_i$ for each $v_i \in V(G)$.  In this paper, we consider interval graphs such that in the underlying interval representation $I$, for each interval $I_i$, its left endpoint $l_i$ and right endpoint $r_i$ are integers.   Note that such an interval representation always exists for any interval graph.   \\
Given an interval graph $G$ with the interval representation $I$, our reduction algorithm constructs the corresponding convex bipartite graph $G^*(X^*,Y^*)$ as follows: $V(G^*)=V_1 \cup V_2$, $V_1=\{y_i ~|~ I_i \mbox{ is an interval in } I \}$, $V_2=\{x_i ~|~ l_i \mbox { for some } I_i \} \cup \{x_j ~|~ r_j \mbox { for some } I_j \}$.   In $G^*$, there is a vertex $y_i$ for each interval $I_i$ and there is a $x$ vertex corresponding to $l_i$ and there is a $x$ vertex corresponding to $r_i$.    $E(G^*)=\{ \{y_i,x_j\} ~|~ y_i \in V_1, x_j \in V_2, x_j \mbox {is in the interval } I_i \}$.  That is, $y_i$ is adjacent to all $x$ vertices that lie in the interval corresponding to $y_i$.   Let $X^*=V_2$ and $Y^*=V_1$.   Clearly, $G^*$ is a convex bipartite graph with convexity on $X^*$.    To solve STREE in interval graphs for the instance $(G,R)$, we solve STREE in convex bipartite graphs for the corresponding instance $(G^*,R^*)$.  
In particular, we map $(G,R=\{z_1,\ldots,z_k\})$ to $(G^*, R^*=\{y_1,\ldots,y_k\})$ such that $y_i \in Y^*$.   Clearly, $R^* \subseteq Y^*$.   Using the dynamic programming presented in Section \ref{rsuby}, we solve $(G^*,R^*)$, and let $S^*$ be the solution to $G^*$.   Let $S=\{v_i ~|~y_i \in S^*\cap Y^*, v_i \in V(G)\}$.   We now claim that $S$ is a Steiner set in $G$ for $R$.   Suppose not, then there exist $v_i, v_j \in R$ such that there is no path between $v_i$ and $v_j$ in the graph induced on $R \cup S$.   This implies that there is no path between $y_i$ and $y_j$ in the graph induced on $R^* \cup S^*$, contradicting the fact that $S^*$ is a Steiner set.   Thus $S$ is a Steiner set in $G$.   Further, $S$ is a minimum Steiner set in $G$. 
	\section{Another Application: Domination in convex bipartite graphs}
	It is known from \cite{damaschke1990domination} that the minimum domination set problem in convex bipartite graphs is polynomial-time solvable.  In this section, we propose an approach which uses STREE in convex bipartite graphs as a black box, and this approach is different from the one reported in \cite{damaschke1990domination}.   Further, we obtain a solution to the domination in convex bipartite graphs in $O(nm)$ time.   Our reduction algorithm takes an instance of domination problem in convex bipartite graphs and maps to the corresponding instance of STREE in convex bipartite graphs.   For $G(X,Y)$ of domination problem, we invoke (i) STREE on $G$ with $R=X$, and (ii) STREE on $G$ with $R=Y$.   Let $D_1$ and $D_2$ be the minimum set of Steiner vertices output by the algorithm when invoked on $(G,R=X)$, and $(G,R=Y)$, respectively.    
	\begin{theorem}
		  $D=D_1\cup D_2$ is a minimum dominating set. 
	\end{theorem}
	\begin{proof}
		Suppose that $D$ is not a minimum dominating set, then there exists a minimum
		dominating set $D'$ such that $|D'|<|D|$.  This implies that $|D'\cap Y|<|D_1|$
		or $|D'\cap X|<|D_2|$.   Further, $D' \cap Y$ is a Steiner set for the case $R=X$ and $D' \cap X$ is a Steiner set for the case $R=Y$.   It contradicts the fact that $D_1$ and $D_2$ are minimum Steiner sets.  Therefore, $D$ is a minimum dominating set.
	\end{proof}
	Remark:  It is shown in \cite{chen2016complexity} that the dominating set problem on comb-convex bipartite graphs is NP-complete.   Since comb-convex bipartite graphs are precisely 1-star caterpillar convex bipartite graphs, the dominating set problem on 1-star caterpillar convex bipartite graphs is NP-complete.  Thus we obtain a dichotomy for the domination in $k$-star caterpillar convex bipartite graphs similar to STREE.
	\\It is important to highlight that Domination on chordal bipartite graphs is NP-complete \cite{muller1987np}.  A micro-level analysis of the reduction instances shows that the instances are a variant of 3-star caterpillar convex bipartite graph: exactly one of the 3-stars is such that one branch is $P_1$ (path of length one) and the other two branches are $P_2$ (path of length two).  In this paper, we strengthen the result of \cite{muller1987np} and show that on 1-star caterpillar convex graphs, the Domination is NP-complete.
	\\\\
	{\bf Conclusions and Directions for Further Research:} 
In this paper, we present an interesting dichotomy: we show that STREE on 0-star caterpillar convex bipartite graphs (convex bipartite graphs) are polynomial-time solvable, whereas STREE on 1-star caterpillar convex bipartite graphs is NP-complete.  Further we show that STREE in interval graphs and Domination in convex bipartite graphs are polynomial-time solvable by using the STREE algorithm for convex bipartite graphs.   Our greedy strategies and dynamic programming based solution exploits the structure
	of convex bipartite graphs which can be used in the study of other combinatorial
	problems such as Steiner path, variants of dominating set, variants of
	Hamiltonicity.    Also, P vs NPC boundary investigation for other combinatorial problems in generalization of convex bipartite graphs would be an interesting direction to explore with.
	\\\\\\
	\bibliographystyle{unsrt}
	\bibliography{references}
\end{document}